\documentclass[runningheads]{llncs}

\newif\iffullversion
\fullversionfalse
\ifdefined\BuildFullVersion\fullversiontrue\fi
\ifdefined\BuildConferenceVersion\fullversionfalse\fi
\newif\ifcompactlayout
\iffullversion\compactlayoutfalse\else\compactlayouttrue\fi
\ifdefined\BuildRestoredVersion\compactlayoutfalse\fi
\ifdefined\BuildCompactVersion\compactlayouttrue\fi
\newif\ifbriefprelim
\iffullversion\briefprelimfalse\else\briefprelimtrue\fi
\ifdefined\BuildRestoredVersion\briefprelimfalse\fi
\ifdefined\BuildEquationOnlyVersion\briefprelimfalse\fi
\compactlayoutfalse

\usepackage[T1]{fontenc}
\usepackage{graphicx}

\usepackage{amsfonts}
\usepackage{amsmath}
\usepackage{amssymb}
\usepackage{mathrsfs}
\usepackage{mathtools}
\usepackage{thmtools}
\usepackage{hyperref}
\hypersetup{bookmarksdepth=3}

\usepackage[capitalize]{cleveref}
\usepackage{hhline}
\usepackage{comment,xspace,xparse}

\usepackage{tikz}
\usetikzlibrary{shapes.geometric, arrows.meta}
\usepackage{pgf}
\usepackage{pdfpages}
\usepackage{scalerel}
\usepackage{nicefrac}

\newcommand{\ket}[1]{| #1\rangle}

\def\leq{\leqslant}\def\le{\leq}
\def\geq{\geqslant}\def\ge{\geq}
\def\emptyset{\varnothing}

\def\oX{\overline{X}}

\usepackage{euscript} 
\usepackage{mathrsfs} 

\def\Jelinek{Jel{\'\i}nek\xspace}
\def\lineg{\mathscr{L}}
\def\sF{\mathscr{F}}
\def\cw{\mathrm{cw}}
\def\rw{\mathrm{rw}}
\def\tw{\mathrm{tw}}
\def\pw{\mathrm{pw}}
\def\chiac{\chi_{\mathrm{ac}}}
\def\chistrong{\chi_{\mathrm{S}}'}
\def\nuac{\nu_{\mathrm{ac}}}
\def\cQ{\mathcal{Q}}
\def\mimw{\mathrm{mimw}}

\newcommand\disfrac[2]{\frac{\displaystyle #1}{\displaystyle #2}}
\DeclareMathOperator*{\bigsquare}{\scalerel*{\square}{\sum}}

\NewDocumentCommand\set{sm}{\IfBooleanTF#1{\{{#2}\}}{\left\{{#2}\right\}}}
\NewDocumentCommand\ceil{sm}{\IfBooleanTF#1{\lceil{#2}\rceil}{\left\lceil{#2}\right\rceil}}
\NewDocumentCommand\floor{sm}{\IfBooleanTF#1{\lfloor{#2}\rfloor}{\left\lfloor{#2}\right\rfloor}}
\NewDocumentCommand\pare{sm}{\IfBooleanTF#1{({#2})}{\left({#2}\right)}}
\NewDocumentCommand\range{smm}{\IfBooleanTF#1{\set*{{#2},\dots,{#3}}}{\set{{#2},\dots,{#3}}}}
\NewDocumentCommand\card{sm}{\IfBooleanTF#1{|{#2}|}{\left|{#2}\right|}}

\newcounter{openquestion}

\let\BEGINproof\proof
\let\ENDproof\endproof
\def\proof{\begingroup\BEGINproof}
\def\endproof{\hfill\qed\ENDproof\endgroup}

\let\tmpOmega\Omega%
\let\tmpTheta\Theta%
\let\tmpDelta\Delta%
\def\Omega{\mathrm{\tmpOmega}}
\def\Theta{\mathrm{\tmpTheta}}
\def\Delta{\mathrm{\tmpDelta}}

\AtBeginDocument{%
  \iffullversion\else
    \setlength{\abovedisplayskip}{7pt plus 2pt minus 2pt}%
    \setlength{\belowdisplayskip}{8pt plus 2pt minus 2pt}%
    \setlength{\abovedisplayshortskip}{0pt plus 2pt}%
    \setlength{\belowdisplayshortskip}{4pt plus 2pt minus 1pt}%
  \fi
}

\begin{document}
\pagestyle{empty}

\title{Entanglement from Expansion: Deterministic Graphs of Large Rank-Width}

\author{Tristan Cam\inst{1,2}\orcidID{0009-0000-3853-5404}\thanks{This work has been supported by Maison du Quantique de Nouvelle-Aquitaine “HybQuant”, as part of the HQI initiative (ANR-22-PNCQ-0002), France 2030 and PEPR EPIQ (ANR-22-PETQ-0007), along with the ANRT Program.}
\and
Cyril Gavoille\inst{1}\orcidID{0000-0003-3671-8607}
\and
Yvan Le~Borgne\inst{1}\orcidID{0009-0008-4023-8677}
\and
Simon Martiel\inst{2}\orcidID{0000-0001-5624-2955}}
\authorrunning{T. Cam, C. Gavoille, Y. Le~Borgne and S. Martiel} 
\titlerunning{Deterministic Graphs of Large Rank-Width}
%
\institute{LaBRI, University of Bordeaux, CNRS, France \and IBM Quantum, IBM France Lab, France}
\maketitle
\thispagestyle{empty}
\pagestyle{empty} 

\vspace{-3mm}
\begin{abstract}
We give a degree-five $n$-vertex graph family with rank-width at least $n/15$. This family yields deterministically constructible graph states of linear entanglement width preparable in at most six CZ layers under all-to-all connectivity.
We also prove that the $k$-dimensional hypercube has rank-width $\Theta(2^k/\sqrt{k})$, the corresponding states providing better entanglement width values for NISQ-accessible numbers of qubits. Our bounds use an expansion-to-rank-width transfer, and a stronger matching construction for Cartesian products.
\keywords{rank-width \and graph state \and hypercube \and Cartesian product}
\end{abstract}

\section{Introduction}


\emph{Entanglement} is known to be essential to the advantage of quantum computing over classical algorithms \cite{vidal2003efficient,jozsa2003role}. In particular, quantum computations that generate only limited entanglement can often be simulated efficiently on a classical computer \cite{perez2006matrix,orus2014practical}, while the presence of sufficiently large multipartite entanglement is necessary for exponential quantum \emph{speedups} \cite{jozsa2003role}. Stabilizer states form a key class of quantum states particularly significant due to
its properties in quantum error correction  \cite{flammia2011direct}. In this work, we focus on \emph{graph states}, a central subclass of stabilizer states that can be represented by undirected graphs. This work is heavily motivated by the pursuit of deterministic graph state preparation with high entanglement \cite{kumabe2024complexity,davies2025preparing,GHH25}. 
This work shows the existence of several families of graph states of entanglement width $\Theta(n)$ (i.e., asymptotically maximum) preparable in constant depth. 
These bounds match the best probabilistic graph state preparation bounds \cite{LLO12,DEK24} while deterministically certifying linear entanglement by design.




\begin{table}[!ht]
\renewcommand{\arraystretch}{1.15}
\caption{\textbf{Families of graph states of near-maximal entanglement.} \\ Entanglement width and entangling gate depth required for circuit preparation in all-to-all connectivity for graph states associated to different families of graphs.}\label{tableres}
\centering
\begin{tabular}{|| c c c c ||} 
 \hline
 Resource graph state & Deterministic? & ~~Entanglement width~~ & Depth \\  
 \hline
 \hline
 Dense Erd\H{o}s-R{\'e}nyi \cite{LLO12} & no & $n/3 - O(1)$ & $O(n)$\\
 \hline 
  Sparse Erd\H{o}s-R{\'e}nyi \cite{LLO12} & no & $\Theta(n)$ & $O(\log n)$\,\footnotemark[1]\\
 \hline
  Hypercube (\textbf{new}) & yes & $\Theta\!\pare{n/\sqrt{\log n}}$ & $O(\log n)$\\
 \hline
 \hline
  Square grid \cite{Jelinek10} & yes & $\sqrt{n}
  - 1$ & $O(1)$\\
 \hline
Random $3$-regular \cite{DEK24}+ Eq.~\eqref{eq:h(G)} & no & $\Omega(n)$  & $O(1)$\\
 \hline
 Explicit Ramanujan (\textbf{new}) & yes & $\geq n/15$ & $O(1)$\\
 \hline
\end{tabular}
\end{table}
The entanglement width of a graph state $\ket{G}$ is known to be the \emph{rank-width} of the corresponding graph $G$ \cite{hein2004multiparty,van2006universal}. 
Rank-width is a graph parameter introduced by Oum and Seymour \cite{oum2006approximating} based on the $\mathbb{F}_2$-rank of matrices describing any cut in the graph. Indeed, while this cut-rank equals bipartite entanglement entropy between two quantum sub-systems \cite{hein2004multiparty}, the rank-width minimizes the maximum of these values over all recursive decompositions.
Moreover, if the maximum degree of the graph is bounded, then the graph state can be prepared efficiently, i.e., by a bounded depth circuit. This motivates the search for deterministic families of graphs that have both bounded degree and near-linear rank-width. Having large rank-width also eliminates the possibility of tensor network contraction methods, one of the most widely used techniques for classical simulation, as a tree tensor network representing $\ket{G}$ exactly requires a bond dimension exponential in $\rw(G)$ across some edge \cite{o2019parameterization}, making our graph states adversarial instances for most classical simulators.
As proved by Oum~\cite{Oum08}, rank-width can be upper bounded by tree-width (plus one). However, providing strong and non-trivial \emph{lower bounds} for a given family of graphs remains a major challenge. 
The most widely cited benchmark \cite{Jelinek10} is the square grid with $n$ vertices, whose rank-width scales precisely as $\sqrt{n}-1$.

\footnotetext[1]{From \cite[Corollary 3.14]{Bollobas01}.}
\subsection{Main technical results}

In this work, we develop a general method to derive rank-width lower bounds for a graph directly from its edge expansion~\cite{Harper66,mohar1989isoperimetric}. Inspired by \Jelinek's method for lower-bounding the rank-width of the grid~\cite{Jelinek10}, we generalize this approach by finding large \emph{acyclic} matchings 
thanks to greedy algorithms and turning those into cut-rank certificates. Our first contribution relies also on the 
\emph{acyclic chromatic index} $\chiac(G)$ as defined by Baste and Rautenbach~\cite{BR18}. 
Writing $\partial X$ for the set of edges of a cut $(X,\oX)$ of $G$ (one end in $X$, the other in $\oX$), we use edge-isoperimetric inequalities of the form $|\partial X| \ge \sF_G(|X|) \ge h(G)\cdot|X|$, where $\sF_G$ is a function of the size of $X$ only, as given by \cite{bollobas1991edge,tillich2000edge}, and $h(G)$ is the edge expansion of $G$. 
Our first theorem can be stated as follows:

\begin{restatable}{theorem}{thmalgo}\label{thm:algo}
Let $G = (V,E)$ be a graph with at least one edge, satisfying $|\partial X|\geq \sF_G(|X|)$ for every $X\subsetneq V$. Then,
\[
    \rw(G) ~\ge~ \min_{\frac{1}{3}\leq\frac{|X|}{|V|}\leq\frac{1}{2}}\frac{|\partial X|}{\chiac(G)} ~\ge~ \min_{p\in[\frac{1}{3},\frac{1}{2}]}\frac{\sF_G\pare{p\,|V|}}{\chiac(G)} ~.
\]
\end{restatable}

\ifbriefprelim\else
As stated formally later, $\chiac(G) \le \Delta(G)^2$, where $\Delta(G)$ denotes the maximum degree of $G$. Thus, the rightmost term in \cref{thm:algo} can be lower-bounded by $\rw(G) \ge \frac{h(G)~\cdot~|V|/3}{\Delta(G)^2}$.
\fi 
Notice that the complete graph $K_n$ has maximal edge expansion $h(K_n) = \Theta(n)$ but minimal rank-width $\rw(K_n) = 1$, signifying that high edge expansion alone does not suffice for guaranteeing large rank-width. On the other hand, the square grid has a much lower edge expansion $h(P_{\sqrt{n}}\square P_{\sqrt{n}}) = \Theta(1/\sqrt{n}\,)$ but its rank-width is $\Theta(\sqrt{n}\,)$ thanks to its low degree. The above inequality confirms the intuition that graphs of both edge expansion bounded away from 0 and bounded maximum degree have large rank-width.

\iffullversion
We show a collection of new lower bounds for a vast array of graph families as seen in~\cref{table}, including on Cartesian products of graphs, high dimension grids and tori, Cayley graphs, hypercubes, and Ramanujan graphs. These later yield the following result.
\fi

Notably, we apply \cref{thm:algo} to expander graphs to obtain rank-width lower bounds linear in the number of vertices, and in particular we obtain the following.
\vspace{-2mm}
\begin{restatable}{theorem}{thmram}\label{thm:ram}
For infinitely many $n$, there is a $5$-regular graph $G_n$ on $n$ vertices, deterministically constructible in polynomial time, such that $~\rw(G_n) \ge \frac{n}{15}\,.$
\end{restatable}

This is obtained by exploiting spectral expansion properties of Ramanujan graphs, a family of expanders, using a neighbour-count inequality due to Brouwer and Haemers \cite{BH11}. 
Our bound counts vertices with neighbours on both cut sides and then extracts a directed acyclic matching (see Section~\ref{sec:ramanujan-five}).

For Cartesian products of connected graphs, we obtain a stronger bound that depends only on the factor orders, without degree or regularity assumptions. For the hypercube it gives a tight asymptotic order. The same construction answers a question about the outer vertex boundary of cuts in Cartesian products, from Diskin, Erde, Kang and Krivelevich \cite[Question 7.3(1)]{diskin2024isoperimetric} (see \cref{dekk-answer}).

\begin{restatable}{theorem}{thmcartesian}\label{thm:cartesian}
For $k\ge1$ and $i\in\range{1}{k}$, let $G_i$ be a connected graph of order $n_i\ge2$. 
Let $G=\bigsquare_{i=1}^kG_i$ be their Cartesian product of order $n=\prod_{i=1}^k n_i$.

\noindent For every cut $(X,\oX)$ of $G$,
\vspace{-3mm}
\begin{equation}\label{eq:c-cut}
 \rho_G(X) ~\ge~ 
 \frac{\sqrt{2}\,|X||\oX|}{n\sqrt{\sum_{i=1}^k n_i^2}}.
\end{equation}
\vspace{-3mm}
Consequently,
\begin{equation}\label{eq:c-rw}
 \rw(G) ~\ge~ \frac{2\sqrt{2}}{9}\cdot
 \frac{n}{\sqrt{\sum_{i=1}^k n_i^2}}.
\end{equation}
\end{restatable}

\noindent Note that Eq.~\eqref{eq:c-cut} holds for every cut, not necessarily balanced.

As these product graphs naturally induce large grids, we show how to adapt a sampling scheme by Haferkamp et al. \cite{haferkamp2020closing} estimating anticoncentration, further indicating classical hardness of the approximate sampling (see \cref{anticoncentration}).
\iffullversion
\begin{table}[ht]
\caption{New rank-width lower bounds for various deterministic graph families.}
\centering
\begin{tabular}{|| c c c | c | c c ||} 
 
 \hline
  Reference & Graph & $G$ & $\rw(G)$ & Order & $\Delta(G)$\\  
 \hhline{||===|=|==||}
    Thm.\ref{thm:algo}, p.\pageref{page:algo} & any connected & $G$ & $\ge\frac{h(G)\, \cdot \, n/3}{\Delta^2}$ & $n$ & $\Delta$\\
  \hline
   Prop.\ref{cor:ramanujan}, p.\pageref{cor:ramanujan} & Ramanujan graph & $G$ & $\Omega\!\pare{\frac{n}{d}}$ & $n$ & $d$\\
 \hline
    Thm.\ref{thm:ram}, p.\pageref{pageram} & Deterministic graph & $G$ & $\geq\frac{n}{15}$ & $n$ & $5$\\
 \hline
  {Thm.\ref{thm:cartesian}, p.\pageref{page:cartesian}}& Cartesian prod. & $\bigsquare_{i=1}^k G_i$ & $\ge\frac{2\sqrt{2}}{9}\cdot\frac{n}{\sqrt{\sum_i n_i^2}}$ & $n=\prod_i n_i$ & $\sum_i\Delta(G_i)$\\
 \hline
  Cor.\ref{cor:cartesian}, p.\pageref{cor:cartesian} & Cartesian pow. & $G^{\square k}$ & $\ge\frac{2\sqrt{2}}{9}\cdot\frac{m^{k-1}}{\sqrt{k}}$ & $m^k$ & $k\Delta(G)$\\
 \hline
  Eq.~\eqref{hypercube}, p.\pageref{hypercube} & Hypercube & $(K_2)^{\square k}$ & $\ge \frac{\sqrt{2}}{9}\cdot\frac{2^k}{\sqrt{k}}$ & $2^k$ & $k$\\
 \hline
\end{tabular}
\label{table}
\end{table}

\fi

\ifbriefprelim

\section{Preliminary Notions}
\label{sec:prelim}

Throughout the paper, all graphs considered will be simple, finite and undirected.

\subsection{Cut-rank and rank-width}

A \emph{cut} in a graph $G=(V,E)$ is a pair of non-empty vertex subsets $(X,\oX)$ such that $X\subsetneq V$, and $\oX = V\setminus X$.
The \emph{edge boundary} of $X$, denoted by $\partial X$, is defined as $\partial X := \set{ (u,v)\in E : u\in X, v\in\oX }$. 

The \emph{bi-adjacency matrix} of a cut $(X,\oX)$ is the 0-1 matrix whose rows are indexed by $X$, columns are indexed by $\oX$, and the entry in row $i$ and column $j$ is equal to 1 if and only if $i$ and $j$ are connected by an edge of $G$.
The \emph{cut-rank} of $X$, denoted by $\rho_G(X)$, is the $\mathbb{F}_2$-rank of the bi-adjacency matrix of $(X,\oX)$.

A tree is \emph{subcubic} if every node has degree $1$ or~$3$. A \emph{rank-decomposition} of a graph $G = (V,E)$ is a pair $(T, L)$ where $T$ is a subcubic tree with at least two nodes, and $L$ is a bijection from $V$ to the set of all leaves of $T$. 

For each edge $e$ of $T$, $T-e$ induces a partition $(A_e, B_e)$ of the leaves of $T$. The \emph{width} of $e$ is $\rho_G(L^{-1}(A_e))$, that is the cut-rank of the vertex set $X\subsetneq V$ such that $L(X) = A_e$. Note that $\rho_G(X) = \rho_G(\oX)$ for all $X\subsetneq V$, so the choice of $A_e$ or $B_e$ does not change the width of $e$.

The \emph{width} of a rank-decomposition $(T, L)$ is the maximum width of edges in $T$. The \emph{rank-width} of a graph $G$, denoted by $\rw(G)$, is the minimum width over all rank-decompositions of $G$. 

It is well-known that (see \cite[Lemma~2.1]{LLO12} for a proof):

\begin{lemma}[Folklore]\label{lemma:rw}
For every graph $G$ with $n\ge 2$ vertices, if for every cut $(X,\oX)$ of $G$ with $\min(|X|,|\oX|)\ge n/3$, we have $\rho_G(X)\ge k$,
then $\rw(G)\ge k$.
\end{lemma}

Intuitively, as every rank-decomposition contains an edge whose removal induces a cut with a \emph{balanced} number of vertices on each side, i.e., such that $\min(|X|,|\oX|)\ge n/3$ 
, then to lower-bound the maximum width over all edges in a decomposition, it is enough to lower-bound the width of edges that induce balanced cuts.

\subsection{Directed acyclic matchings}\label{sec:ac}

The \emph{cut-graph} with respect to a cut $(X,\oX)$ in $G$, written $G[\partial X]$ or simply $C$ when not ambiguous, is the graph with edge set $\partial X$ and its incident vertices, i.e., the bipartite subgraph of $G$ composed of the union of all the edges of $\partial X$.

A \emph{matching} $M$ in a graph $G$ is an independent set of edges, i.e., such that no two edges of $M$ have a common endpoint. 
We will refer to matchings in the cut-graph $G[\partial X]$ as $(X,\oX)$-matchings.

With any $(X,\oX)$-matching $M$ we associate an edge-adjacency digraph $A_{M,X}$ whose vertices correspond to the edges of $M$, and if $e = \set{x,y}$ and $e' = \set{x',y'}$ are two distinct edges of $M$ such that $x,x'\in X$ and $y,y'\in \oX$, then $A_{M,X}$ has a directed edge from $e$ to $e'$ if and only if $y$ is adjacent to $x'$ in $G$. 

We say that a $(X,\oX)$-matching $M$ is \emph{directed acyclic}\footnote[2]{The adjective \emph{directed} refers to the use of $A_{M,X}$. It is added for clarity and does not appear in Jel\'inek's original work.} if $A_{M,X}$ does not contain any directed cycle; not even a directed cycle of length two, corresponding to edges of $M$ in the same \emph{four-cycle} (an edge subset forming a $C_4$).

\Jelinek proposed the following method to lower-bound the cut-rank.

\begin{lemma}[\Jelinek~{\cite[Lemma~2]{Jelinek10}}]\label{lemma:jelinek}
Let $M$ be a directed acyclic $(X,\oX)$-matching of $G$. Then, $\rho_G(X) \ge |M|$.
\end{lemma}

\subsection{Colouring-based matching extraction}

A matching $M$ of $G$ is \emph{acyclic} if $G[V(M)]$ is acyclic, i.e., the graph induced by the endpoints of the edges of $M$ is a forest. The matching $M$ is \emph{induced} if $G[V(M)]$ is the matching itself, i.e., a forest composed of an independent set of edges.

It is important to notice that a directed acyclic $(X,\oX)$-matching $M$ of $G$, as defined by \Jelinek in \cref{sec:ac}, is not necessarily an acyclic matching in $G[\partial X]$ as defined just above.
However, the reverse is true: if $M$ is an acyclic matching of $G[\partial X]$, then it is a directed acyclic $(X,\oX)$-matching.

The \emph{acyclic chromatic index} of $G$, denoted by $\chiac(G)$, is the smallest number $k$ such that $G$ has a partition of its edges into $k$ sets such that each one is an acyclic matching. 
The \emph{acyclic matching number} of $G$, denoted by $\nuac(G)$, is the largest number of edges of an acyclic matching contained in $G$. Clearly, if $G$ has a total of $m>0$ edges, $\nuac(G) \ge m/\chiac(G)$. We will often use the following upper bound.

\begin{lemma}[Baste and Rautenbach~{\cite[Thm.~2]{BR18}}]\label{lemma:acyclic}
Let $G$ be a graph of maximum degree $\Delta$. Then, $\chiac(G)\leq \Delta^2$. For $\Delta\ge1$, equality holds if and only if a connected component is $K_{\Delta,\Delta}$, the complete bipartite graph with parts of size $\Delta$.

\end{lemma}

\subsection{Isoperimetric inequalities}

The \emph{edge expansion} $h(G)$ of $G$ (also called the \emph{Cheeger constant}) is then defined as
$
    h(G) = \min_{X,\oX\subsetneq V}\frac{|\partial X|}{\min(|X|,|\oX|)}\,, ~ |V|\ge2
$. 
Since $\partial X = \partial \oX$, we have by symmetry $h(G) = \min_{0<|X|\le |V|/2} |\partial X|/|X|$. In the following, we will be interested in graphs whose edge expansion is bounded away from 0.
We will then say that, given such a graph $G$ and a cut  $(X,\oX)$ (where $X$ is chosen to be the subset satisfying $|X|\le |V|/2$), we can lower-bound the size of the edge boundary as
$|\partial X| \ge \sF_G(|X|)\ge h(G)\cdot|X|$, for some function $\sF_G$, a so-called \emph{edge-isoperimetric inequality} of $G$ (as described by Tillich \cite{tillich2000edge} for instance).
\else
\section{Preliminary Notions}
\label{sec:prelim}

Throughout the paper, all graphs considered will be simple, finite and undirected.

\subsection{Cut-rank and rank-width}

A \emph{cut} in a graph $G=(V,E)$ is a pair of non-empty subsets $(X,\oX)$ such that $X\subsetneq V$, and $\oX = V\setminus X$.
The \emph{bi-adjacency matrix} of a cut $(X,\oX)$ is the 0-1 matrix whose rows are indexed by $X$, columns are indexed by $\oX$, and the entry in row $i$ and column $j$ is equal to 1 if and only if $i$ and $j$ are connected by an edge of $G$.
The \emph{cut-rank} of $X$, denoted by $\rho_G(X)$, is the $\mathbb{F}_2$-rank of the bi-adjacency matrix of $(X,\oX)$.

A tree is \emph{subcubic} if every node has degree $1$ or~$3$. A \emph{rank-decomposition} of a graph $G = (V,E)$ is a pair $(T, L)$ where $T$ is a subcubic tree with at least two nodes, and $L$ is a bijection from $V$ to the set of all leaves of $T$. For each edge $e$ of $T$, $T-e$ induces a partition $(A_e, B_e)$ of the leaves of $T$. The \emph{width} of $e$ is $\rho_G(L^{-1}(A_e))$, that is the cut-rank of the vertex set $X\subsetneq V$ such that $L(X) = A_e$. Note that $\rho_G(X) = \rho_G(\oX)$ for all $X\subsetneq V$, so the choice of $A_e$ or $B_e$ does not change the width of $e$.

The \emph{width} of a rank-decomposition $(T, L)$ is the maximum width of edges in $T$. The \emph{rank-width} of a graph $G$, denoted by $\rw(G)$, is the minimum width over all rank-decompositions of $G$. (If $G$ has less than two vertices, then $G$ has no rank-decompositions; in this case we say that $G$ has rank-width 0.)




It is well-known that (see \cite[Lemma~2.1]{LLO12} for a proof):

\begin{lemma}[Folklore]\label{lemma:rw}
For every graph $G$ with $n\ge 2$ vertices, if for every cut $(X,\oX)$ of $G$ with $\min(|X|,|\oX|)\ge n/3$, we have $\rho_G(X)\ge k$, 
then $\rw(G)\ge k$.
\end{lemma}

Intuitively, as every rank-decomposition 
contains an edge whose removal induces a cut with a \emph{balanced} number of vertices on each side, i.e., such that $\min(|X|,|\oX|)\ge n/3$ (\emph{intermediate value theorem} argument), then to lower-bound the maximum width over all edge in a decomposition, it is enough to lower-bound the width of edges that induce balanced cuts. Therefore, when investigating the width of all decompositions, the most common approach is to find the minimum width over any balanced cut.

\subsection{Directed acyclic matchings}\label{sec:ac}

The \emph{cut-graph} with respect to a cut $(X,\oX)$ in $G$, written $G[\partial X]$ or simply $C$ when not ambiguous, is the subgraph induced by $\partial X$, i.e., the bipartite subgraph of $G$ composed of the union of all the edges of $\partial X$. 

A \emph{matching} $M$ in a graph $G$ is an independent set of edges, i.e., such that no two edges of $M$ have a common endpoint. For convenience, we denote by $V(M)$ the set of endpoints of all the edges of $M$. We will refer to matchings in the cut-graph $G[\partial X]$ as $(X,\oX)$-matchings.

With any $(X,\oX)$-matching $M$ we associate an edge-adjacency digraph $A_{M,X}$ whose vertices correspond to the edges of $M$, and if $e = \set{x,y}$ and $e' = \set{x',y'}$ are two distinct edges of $M$ such that $x,x'\in X$ and $y,y'\in \oX$, then $A_{M,X}$ has a directed edge from $e$ to $e'$ if and only if $y$ is adjacent to $x'$ in $G$. We say that a $(X,\oX)$-matching $M$ is \emph{directed acyclic}\footnote{The precision \emph{directed} relating to the use of $A_{M,X}$ is added for clarity and does not appear in Jel\'inek's original work.} if $A_{M,X}$ does not contain any directed cycle (not even a directed cycle of length two).


\Jelinek proposed the following method to lower-bound the cut-rank.

\begin{lemma}[\Jelinek~{\cite[Lemma~2]{Jelinek10}}]\label{lemma:jelinek}
Let $M$ be a directed acyclic $(X,\oX)$-matching of $G$. Then, $\rho_G(X) \ge |M|$.
\end{lemma}




\subsection{colouring-based matching extraction}

\iffullversion
The \emph{line graph} of $G=(V,E)$, denoted by $\lineg(G)$, is the graph that has as vertex set $E$ with two vertices adjacent if and only if their corresponding edges in $V$ are incident.

The \emph{$k$-th power of $G$}, denoted by $G^k$, is a graph that has the same vertex set as $G$, and two vertices are adjacent if and only if they are within distance at most $k$ in $G$.


\fi
The \emph{Cartesian product} of two graphs $G$ and $H$, denoted by $G \square H$, is the graph with vertex set $V(G) \times V(H)$, and edge set $$\set{(u, x)(v, y) : (u,v)\in E(G)\land x = y\text{, or }(x,y)\in E(H)\land u = v }$$
Given graphs $G_1,\dots,G_k$, we write their Cartesian product as $\bigsquare_{i=1}^k G_i$ or simply 
${G}^{\square k}$ if $G_1 = \cdots = G_k = G$.

Given a graph $G=(V,E)$, a \emph{proper edge colouring} is a mapping $c:E\rightarrow\mathbb{N}$ satisfying $c[(u,v)]\neq c[(v,w)],~\forall (u,v)\neq(v,w)\in E.$

A \emph{strong edge-colouring} of $G$ is a proper edge-colouring such that no edge is adjacent to two edges of the same colour. 
The \emph{strong chromatic index} of $G$, often denoted as $\chistrong(G)$ or $\chi_2'(G)$, is the least integer $k$ such that there exists a strong edge-colouring of $G$ using $k$ colours.

\iffullversion
Notice that, if we denote by $\chi_2(G)$ the \emph{distance-2 chromatic number} of $G$, i.e., the smallest number of colours required to colour all the vertices of $G$, such that no two vertices within distance $2$ share a colour, then the strong edge-colouring problem can related to colouring powers of graphs as follows $\chistrong(G) = \chi_2(\lineg(G)) = \chi(\lineg(G)^2).$

Thus, it is clear that any colour in a strong edge-colouring of $G$ corresponds to a stable in $\lineg(G)^2$, which precisely corresponds to an induced matching.
As, if $G$ has a strong edge-colouring using $\chistrong(G)$ colours and $|E|=m$, there exists a colour that is shared by at least $\frac{m}{\chistrong(G)}$ edges, this means that there will always exist an induced matching $M\subseteq E$ such that $|M|\geq\frac{m}{\chistrong(G)}$ (and we can write the induced matching number $\nu_{\mathrm{S}}(G)\geq \frac{m}{\chistrong(G)}$). Note that if $G$ can be strongly edge-coloured using $\chistrong(G)$ coloured, then the same goes for the subgraph $G[\partial X]$, for any $X\subseteq V$ (such as a cut-graph, that we will use repeatedly).


\fi

A matching $M$ of $G$ is \emph{acyclic} if $G[V(M)]$ is acyclic, i.e., the graph induced by the endpoints of the edges of $M$ is a forest. The matching $M$ is \emph{induced} if $G[V(M)]$ is the matching itself, i.e., a forest composed of an independent set of edges.

It is important to notice that a directed acyclic $(X,\oX)$-matching $M$ of $G$, as defined by \Jelinek in \cref{sec:ac}, is not necessarily an acyclic matching in $G[\partial X]$ as defined just above. 
However, the reverse is true: if $M$ is an acyclic matching of $G[\partial X]$, then it is a directed acyclic $(X,\oX)$-matching.
In particular, \cref{lemma:jelinek} holds for directed acyclic, acyclic, or induced matchings, i.e., $\rho_G(X) \ge |M|$ for any directed acyclic, acyclic, or induced matching $M$. 

The \emph{acyclic chromatic index} of $G$, denoted by $\chiac(G)$, is the smallest number $k$ such that $G$ has a partition of its edges into $k$ sets such that each one is an acyclic matching. The \emph{acyclic matching number} of $G$, denoted by $\nuac(G)$, is the largest number of edges of an acyclic matching contained in $G$. Clearly, if $G$ has a total of $m>0$ edges, $\nuac(G) \ge m/\chiac(G)$.

We will frequently use the following upper bound on $\chiac(G)$.

\begin{lemma}[Baste and Rautenbach~{\cite[Theorem~2]{BR18}}]\label{lemma:acyclic}
Let $G$ be a graph of maximum degree $\Delta$, then $\chiac(G)\leq \Delta^2$, with equality, for $\Delta\ge1$, if and only if a connected component is $K_{\Delta,\Delta}$, the complete bipartite with parts of size $\Delta$.
\end{lemma}

\iffullversion
Note that if $G$ has maximum degree $\Delta\ge2$ and girth at least~$5$, then no subgraph $H$ of $G$ (for instance a cut-graph) has a component $K_{\Delta,\Delta}$, and thus satisfies $\chiac(H)\leq\Delta^2-1,$ since restricting an edge colouring preserves its matching and forest conditions.



\fi

\subsection{Isoperimetric inequalities}

Let $G = (V,E)$ be a graph with a cut $(X,\oX)$. The \emph{edge boundary} of $X$, denoted by $\partial X$, is defined as
$\partial X := \set{ uv \in E : u\in X, v\in \oX }$.
The \emph{edge expansion} $h(G)$ of $G$ (also called \emph{Cheeger constant}) is then defined as
\[
    h(G) ~=~ \min_{\emptyset\ne X\subsetneq V}\frac{|\partial X|}{\min(|X|,|\oX|)} ~.
\]
Since $\partial X = \partial \oX$, we have by symmetry $h(G) = \min_{0<|X|\le |V|/2} |\partial X|/|X|$. In the following, we will be interested in graphs whose edge expansion is bounded away from 0.
We will then say that, given such a graph $G$ and a cut  $(X,\oX)$ (where $X$ is chosen to be the subset satisfying $|X|\le |V|/2$), we can lower-bound the size of the edge boundary as
$|\partial X| \ge \sF_G(|X|)\ge h(G)\cdot|X|$, for some function $\sF_G$, so-called \emph{edge-isoperimetric inequality} of $G$ (as described by Tillich \cite{tillich2000edge} for instance).

\fi
\section{Lower-Bounding Rank-Width from Expansion}
\label{sec:rw}

We will use isoperimetric inequalities and the acyclic chromatic index to lower-bound rank-width.

\subsection{A general isoperimetric transfer theorem}

\thmalgo*\label{page:algo}



\begin{proof}
Every acyclic matching of $C=G[\partial X]$ is directed acyclic in particular, so \cref{lemma:jelinek} gives $\rho_G(X)\ge\nuac(C)$. 
Partition $E(G)$ into $\chiac(G)$ acyclic matchings. One of them contains at least $|\partial X|/\chiac(G)$ crossing edges, thus $\rho_G(X)\ge\nuac(C)\ge|\partial X|/\chiac(G)$. 
Considering only balanced cuts in virtue of \cref{lemma:rw} and taking $X$ to be the smaller side, we use an edge-isoperimetric inequality $|\partial X|\ge\sF_G(|X|)$ to get the desired
$
    \rw(G) \geq \min_{\frac{1}{3}\leq\frac{|X|}{|V|} \leq \frac{1}{2}}\frac{|\partial X|}{\chiac(G)} ~\geq~ \min_{p\in[\frac{1}{3},\frac{1}{2}]}\frac{\sF_G\pare{p\,|V|}}{\chiac(G)} ~.
$\end{proof}

\vspace{-2mm}
\noindent Since $\sF_G(|X|) \ge h(G)\cdot |X|$, and $\chiac(G)\leq\Delta(G)^2$ (\cref{lemma:acyclic}), we have:

\vspace{-2mm}
\begin{equation}\label{eq:h(G)}
    \rw(G) ~\ge~ \min_{p\in[\frac{1}{3},\frac{1}{2}]} \disfrac{\sF_G(p\,|V|)}{\Delta(G)^2} ~\geq~ \frac{h(G) \cdot |V|/3}{\Delta(G)^2} ~.
\end{equation}

\iffullversion
The preceding bound uses the global parameter $\chiac(G)$. We can also greedily retain the degrees in the cut-graph $C=G[\partial X]$ itself. 
For $v\in V(G)$, let $\deg_{C}(v)$ be the number of edges of $\partial X$ incident with $v$. Thus it counts neighbours in $\oX$ when $v\in X$, and neighbours in $X$ when $v\in\oX$.
Edges whose endpoints have small cut degrees contribute more to the following bound.

\begin{proposition}[Cut-rank from cut degrees]\label{prop:reciprocal-cut-degrees}
For every cut $(X,\oX)$ of a graph $G$,
\begin{equation}\label{eq:reciprocal-cut-degrees}
 \rho_G(X)\ge
 \sum_{(u,v)\in\partial X}
 \frac{1}{\deg_{C}(u)\deg_{C}(v)}.
\end{equation}
\end{proposition}

\begin{proof}
Choose arbitrary vertex orders for $X$ and $\oX$. Retain a crossing edge $(u,v)$, with $u\in X$ and $v\in\oX$, if $v$ is the first neighbour of $u$ in the order on $\oX$, and $u$ is the first neighbour of $v$ in the order on $X$. These edges form a matching $M$. Index them as $(u_i,v_i)$, with $u_i\in X$ and $v_i\in\oX$, in increasing order of $v_i$. Since $(u_i,v_j)$ is not an edge for $j<i$, the selected cut submatrix is upper triangular with diagonal entries one, giving $\rho_G(X)\ge|M|$ as in \cref{lemma:jelinek}.

Each crossing edge $(u,v)$ is retained with probability $1/(\deg_{C}(u)\deg_{C}(v))$. Taking expectations proves the bound.\end{proof}

In particular, this proposition can be applied to balanced each cut, just as in \cref{thm:algo}. It is useful when we can control the actual cut degrees: replacing both cut degrees in each summand by $\Delta(G)$ gives again $\rho_G(X)\ge |\partial X|/\Delta(G)^2$.
\fi




Although its consequence for deterministic constructions is never stated explicitly, a lower bound of similar asymptotical strength can be derived from the proof of \cite[Theorem~4.4]{LLO12} that lower-bounds the rank-width of sparse random graphs that retain a \emph{giant component}, applying a nonzero-entry counting argument to the biadjacency matrix to lower-bound its rank in terms of $1/\Delta(G)^2$.

In general, this last inequality shows that graphs that are bounded-degree expanders have large rank-width (and having both $\Omega(1)$ edge expansion and $O(1)$ maximum degree $\Delta(G)$ implies $\Theta(n)$ rank-width), as developed in the following.

\subsection{Explicit bounded-degree graphs of large rank-width}\label{sec:spectral}

Throughout this subsection, let $G$ be a connected $d$-regular graph with $n$ vertices and $d\geq1$. The Cheeger inequalities \cite{alon1986eigenvalues} give bounds on the edge expansion using the spectral gap $\lambda=d-\lambda_2$:
$\frac{\lambda}{2}\leq h(G)\leq \sqrt{2d\lambda}$,
where $\lambda_1\geq\lambda_2\geq\dots\geq\lambda_n$ are the real eigenvalues of the adjacency matrix of $G$.

\iffullversion
By standard spectral graph theory, the trivial eigenvalue of the adjacency matrix of G is $\lambda_1 = d$ and the first non-trivial eigenvalue is $\lambda_2$. If $G$ is connected, then $\lambda_2 < d$ and it is known that $\lambda_n = -d$ if and only if $G$ is bipartite. Note that the lower bound is tight for hypercubes since $h(\cQ_k)=1$ and the spectral gap is $\lambda=2$.
\fi

We use the following form of the Cheeger inequalities.

\begin{lemma}[Brouwer and Haemers~{\cite[Proposition~4.4.2]{BH11}}]\label{EML}
Let $G$ be a connected $d$-regular graph on $n\geq2$ vertices, other than $K_2$ and $K_3$.

Then, $\frac{d-\lambda_2}{2}\leq h(G)\leq\sqrt{d^2-\lambda_2^2}$.
\end{lemma}
The proof of the lower inequality uses, for a nontrivial cut $(X,\overline{X})$,
$
 \frac{n|\partial X|}{|X||\overline{X}|}\geq \lambda
$.
Rearranging gives the cut estimate:
\begin{equation}\label{eq:spectral-cut-bound}
 |\partial X| ~\ge~ \lambda\frac{|X||\overline{X}|}{n} ~.
\end{equation}

\noindent Put $\lambda(G)=\max_{i\neq1}|\lambda_i|$. Then $d-\lambda(G)\leq d-\lambda_2=\lambda$, inequality can be strict.

Let us consider the best expander graph class from the spectral point of view. A connected nonbipartite graph $G$ that is $d$-regular is said to be \emph{Ramanujan} when $\lambda(G)$ satisfies $\lambda(G)\leq2\sqrt{d-1}$. This bound is asymptotically sharp: the Alon-Boppana bound \cite{nilli1991second} states that for every fixed $d\geq3$ and $\epsilon>0$, there exists $n_0$ such that all $d$-regular graphs $G$ with at least $n_0$ vertices satisfy $\lambda(G)\geq\lambda_2\geq2\sqrt{d-1}-\epsilon$. This means that Ramanujan graphs attain the optimal asymptotic spectral bound for fixed-degree families, and their spectral gap satisfies
\[
 \lambda=d-\lambda_2
 \geq d-\lambda(G)
 \geq d-2\sqrt{d-1}.
\]

\begin{proposition}[Expander graphs]\label{cor:ramanujan}
    Let $G$ be a connected $d$-regular graph on $n$ vertices that is not $K_{d,d}$, with $d\geq2$ and spectral gap $\lambda$. Then, $\rw(G) \ge \frac{2}{9}\cdot\frac{n\lambda}{d^2-1}\,$.
\end{proposition}

\begin{proof}
For a balanced cut $(X,\overline{X})$, both sides contain at least
$\frac{n}{3}$ vertices, so $|X||\overline{X}|\geq\frac{2n^2}{9}$.
By Eq.~\eqref{eq:spectral-cut-bound}, every balanced cut therefore has at least
$\frac{2n\lambda}{9}$ crossing edges.
Applying \cref{thm:algo} and the bound $\chiac(G)\leq d^2-1$ from
\cref{lemma:acyclic}, we obtain
$
 \rw(G)\geq\frac{2n\lambda}{9\chiac(G)}
 \geq\frac{2n\lambda}{9(d^2-1)}\,.
$ Then a nonbipartite Ramanujan graph $G$ with $d\geq3$ satisfies
$\rw(G)\geq n\cdot\frac{2}{9}\cdot\frac{d-2\sqrt{d-1}}{d^2-1}\,.$

\noindent Therefore, $\rw(G)=\Omega\pare{\frac{n}{d}}$ and in particular $\rw(G)=\Theta(n)$ when $d$ is fixed.\end{proof}


To find a numerical lower bound on the rank-width of deterministic graphs of bounded degree from \cref{cor:ramanujan}, we can use the Ramanujan graphs constructed by Lubotzky--Phillips--Sarnak (LPS) \cite{lubotzky1988ramanujan} and get an infinite family of $6$-regular graphs whose rank-width is at least
 $n\cdot\frac{12-4\sqrt{5}}{315} > \frac{n}{104}\,.$

\subsection{A degree-five Ramanujan family}
\label{sec:ramanujan-five}

We now obtain an explicit constant from a degree-five Ramanujan family. The proof counts vertices with neighbours on both sides of a cut, then uses the absence of four-cycles to extract a large directed acyclic matching. 

\begin{proposition}\label{prop:ramanujan-five}
Let $G$ be a connected nonbipartite $5$-regular Ramanujan graph on $n$
vertices, with no four-cycle. For every cut $(X,\oX)$,
 $~\rho_G(X)\ge\frac{3|X||\oX|}{10n}$.
\end{proposition}
Applying this estimate to Morgenstern's explicit family \cite{Mor94} gives our rank-width bound. The proof of \cref{prop:ramanujan-five} and the next theorem are given hereafter. 

\thmram*\label{pageram}

In a cut-graph $C=G[\partial X]$, for $v\in V(G)$, let the cut degree $\deg_{C}(v)$ be the number of edges of $\partial X$ incident with $v$. 
For a $\Delta$-regular graph, vertices with neighbours on both cut sides are exactly those with cut degree between $1$ and $\Delta-1$. 
The following lemma converts their number into a directed acyclic matching, using only a bound on the cut degrees and the absence of four-cycles.

\begin{lemma}\label{lemma:cut-degree-matching}
Let $G$ be a graph of maximum degree $\Delta$ with no four-cycle and $C$ be the cut-graph w.r.t. a cut $(X,\oX)$ in $G$. Let $S := \set{v\in V(C): 1\le \deg_{C}(v) < \Delta }$. Then, there is a directed acyclic $(X,\oX)$-matching $M$ such that
$\,|M|\ge\frac{|S|}{\Delta}$.
\end{lemma}

\iffullversion
No regularity is assumed in this lemma.
\fi

\ifbriefprelim

\begin{proof}[Sketch for \cref{lemma:cut-degree-matching}] Initialize an empty matching $M$, let the current graph be the cut-graph $C$ with a fixed arbitrary order on its vertices to resolve choices. While it still contains vertices, we select edges following two simple rules: if there is a vertex $v\in\oX$ of cut degree 1 in the current graph, select its only incident edge $(u,v)$ into $M$, then remove $u$ and all its neighbours from the current graph. Otherwise, select the first cut-edge $(u,v)$ with $v\in\oX$ into $M$ according to the fixed order, then remove $v$ and all its neighbours from the current graph. The absence of four-cycles ensures no vertex is left isolated, so every vertex of $S$ is eventually counted. At most $\Delta$ vertices of $S$ are counted per edge selected for $M$. Finally, each selected endpoint in $\oX$ loses all its remaining neighbours, so every arc of $A_{M,X}$ is ordered. This gives $|M|\ge |S|/\Delta$ and directed acyclicity. The complete proof appear in \ifcompactlayout
the full version \cite{CGLBM26}.
\else
Appendix~\ref{app:proof-cut-degree-matching}.\fi\end{proof}

\else

\begin{proof}
We start with the whole cut-graph $C = G[\partial X]$, including its vertices of cut degree $\Delta$ which do not belong to $S$, and record the selected edges in $M$.  At each step, we select an edge and delete its endpoints together with some neighbours, retaining all crossing edges between surviving vertices. We count only deletions from $S$: although a step may delete $\Delta+1$ vertices in total, we will show that at most $\Delta$ of them belong to $S$. 
The set $S$ remains fixed by the original cut degrees. During the procedure, neighbours and degrees refer to the graph remaining after these deletions.

The initial cut-graph has no isolated vertices. The two rules below preserve this property, ensuring that every vertex of $S$ is deleted in a step selecting an edge. While an edge remains, apply one of the following rules, using a fixed arbitrary vertex ordering to resolve choices.

\emph{If the remaining part of $\oX$ contains a vertex $v$ of degree one}, let $u\in X$ be its neighbour. Select $(u,v)$ and delete $u$ together with all degree-one vertices of $\oX$ adjacent to $u$. Any surviving neighbour of $u$ had degree at least two and therefore still has a neighbour. The deleted vertices of $\oX$ were adjacent only to $u$, so no surviving vertex of $X$ loses a neighbour.

\emph{Otherwise, every remaining vertex of $\oX$ has degree at least two.} Choose an edge $(u,v)$ with $u\in X$ and $v\in\oX$, select it, and delete $v$ together with all its neighbours.  Every remaining vertex of $\oX$ distinct from $v$ shares at most one neighbour with $v$, since two common neighbours would give a four-cycle. Each of them therefore loses at most one neighbour and remains nonisolated. Surviving vertices of $X$ were not adjacent to $v$, so they lose no neighbour.

Neither rule creates an isolated vertex, so every vertex of $S$ is eventually deleted. Each step deletes one vertex ($u$ in the first rule, $v$ in the second) and some of its neighbours: at most $\Delta+1$ vertices. If a step deletes $\Delta+1$ vertices, then $u$ in the first rule, or $v$ in the second, has degree $\Delta$ immediately before the step. Deletion cannot increase degrees, and the original cut-graph had maximum degree at most $\Delta$, so its original cut degree was also $\Delta$. This vertex does not belong to $S$, whose vertices have original cut degree at most $\Delta-1$. Such a step therefore deletes at most $\Delta$ vertices of $S$. Any smaller step deletes at most $\Delta$ vertices in total. Thus every selected edge accounts for at most $\Delta$ deleted vertices of $S$, giving $|S|\le\Delta|M|$ and hence $|M|\ge \frac{|S|}{\Delta}$.

The selected edges form a matching because their endpoints are deleted immediately. In the first rule, the target $v\in\oX$ has only remaining neighbour $u$, which gets deleted. In the second rule, all its remaining neighbours are deleted. Any endpoint $w\in X$ of an edge selected later by the procedure is present when $(u,v)$ is selected. If $w$ were adjacent to $v$ in $G$, either rule would delete $w$. Thus no earlier selected $v$ is adjacent to a later selected endpoint in $X$. Every arc in $A_{M,X}$ therefore goes from a later selected edge to an earlier one, so $A_{M,X}$ is acyclic. Finally, \cref{lemma:jelinek} gives $\rho_G(X)\ge |M|$.\end{proof}
\fi

\ifbriefprelim

\begin{proof}[Sketch for \cref{prop:ramanujan-five}]
Fix a cut $(X,\oX)$ and let $S$ count the boundary vertices with neighbours on both cut sides. Since $G$ is $5$-regular, these are exactly the vertices with cut degree between one and four. \cref{lemma:cut-degree-matching} already gives a directed acyclic matching of size at least $\frac{|S|}{5}$.
Now to bound $|S|$, remark that its elements are exactly the middle vertices of paths of length two with endpoints on opposite sides of the cut: each uses one crossing edge and one edge staying within a side, satisfying the degree requirements. To account for multiplicity, there are only two cases for the number of neighbours a vertex of $S$ has on each cut side: either one and four, or two and three. Then, it is the middle vertex of at most $\max\{1\cdot4,\,2\cdot3\}=6$ such paths. The Ramanujan spectral bound and the second moment neighbour-count inequality of Brouwer and Haemers \cite[Proposition~4.3.1]{BH11} give at least $9|X||\oX|/n$ such paths. Then, \cref{lemma:jelinek} gives the cut-rank bound
$\, \rho_G(X)\ge\frac{|S|}{5}
 \ge\frac{1}{5}\cdot\frac{1}{6}\cdot\frac{9|X||\oX|}{n}
 =\frac{3|X||\oX|}{10n}$. 
 
 \noindent The complete proof is in 
\ifcompactlayout
the full version \cite{CGLBM26}.
\else
Appendix~\ref{app:proof-ramanujan-five}.
\fi\end{proof}

\else
To apply \cref{lemma:cut-degree-matching}, we need many boundary vertices with cut degree between one and four. Since $G$ is $5$-regular, being incident to a crossing edge ensures cut degree at least one, while having an incident edge staying within the same cut side ensures cut degree at most four. We therefore count pairs of edges sharing a vertex, one crossing the cut and one staying within a side. These pairs are exactly the paths of length two with endpoints on opposite sides of the cut (and middle vertex in either side, but no alternating path since endpoints would be on the same cut side). The proof below uses a spectral estimate of this count to show that many vertices have cut degree between one and four.

\begin{proof}[\cref{prop:ramanujan-five}]
Fix a cut $(X,\oX)$, its cut-graph $C=G[\partial X]$, and let $S=\{v\in V(C):1\le \deg_{C}(v)<5\}$. Write $N_G(v)$ for the set of neighbours of $v$.
All adjacency eigenvalues other than five have absolute value at most four, so the neighbour-count inequality of Brouwer and
Haemers \cite[Proposition~4.3.1]{BH11} gives
\ifcompactlayout

\begingroup
\allowdisplaybreaks[1]
\begin{align}
&\sum_{v\in V(G)} \left(|N_G(v)\cap X|-\frac{5|X|}{n}\right)^2 \le16\frac{|X||\oX|}{n}.\label{eq:ramanujan-five-spectral}\\
\shortintertext{Here $|N_G(v)\cap X|$ equals $\deg_{C}(v)$ when $v\in\oX$, and $5-\deg_{C}(v)$ when $v\in X$. Each vertex of $X$ belongs to $N_G(v)\cap X$ for exactly five vertices $v$, namely its five neighbours. Thus}
&\sum_{v\in V(G)}|N_G(v)\cap X|=5|X|\label{identity-X}\\
\shortintertext{We now expand each square and sum its three terms. The constant term occurs once for each of the $n$ vertices, giving}
&\sum_v\left(|N_G(v)\cap X|-\frac{5|X|}{n}\right)^2\notag\\*
&\quad=\sum_v|N_G(v)\cap X|^2-\frac{10|X|}{n}\sum_v|N_G(v)\cap X|+n\left(\frac{5|X|}{n}\right)^2\notag\\*
&\quad=\sum_v|N_G(v)\cap X|^2-\frac{50|X|^2}{n}+\frac{25|X|^2}{n}=\sum_v|N_G(v)\cap X|^2-\frac{25|X|^2}{n}\label{expansion-square}\\
\shortintertext{Rearranging Eq.~\eqref{expansion-square}, we get}
&\sum_v|N_G(v)\cap X|^2 =\frac{25|X|^2}{n} +\sum_v\left(|N_G(v)\cap X|-\frac{5|X|}{n}\right)^2\label{expansion-square2}\\
\shortintertext{Summing over middle vertices counts each such path once. Their total number therefore satisfies}
&\sum_v \deg_{C}(v)\bigl(5-\deg_{C}(v)\bigr) =\sum_v |N_G(v)\cap X|\bigl(5-|N_G(v)\cap X|\bigr)\notag\\*
&\quad=5\sum_v|N_G(v)\cap X|-\sum_v|N_G(v)\cap X|^2=25|X|-\sum_v|N_G(v)\cap X|^2\label{path-count}\\
\shortintertext{where the last equality uses Eq.~\eqref{identity-X}. Substituting Eq.~\eqref{expansion-square2} into the path count Eq.~\eqref{path-count}, and using $n-|X|=|\oX|$, we obtain}
&\sum_v \deg_{C}(v)\bigl(5-\deg_{C}(v)\bigr)
 =25|X|-\frac{25|X|^2}{n}
   -\sum_v\left(|N_G(v)\cap X|-\frac{5|X|}{n}\right)^2\notag\\*
&\quad=\frac{25|X||\oX|}{n}
   -\sum_v\left(|N_G(v)\cap X|-\frac{5|X|}{n}\right)^2\notag\\*
&\quad\ge\frac{25|X||\oX|}{n}-\frac{16|X||\oX|}{n}=\frac{9|X||\oX|}{n}.\label{eq:ramanujan-five-pairs}\\
\shortintertext{The last inequality is Eq.~\eqref{eq:ramanujan-five-spectral}. Every such path has its middle vertex in $S$ by construction. There are only two cases for the number of neighbours a vertex of $S$ on each cut side: either one and four, or two and three. So it is the middle vertex of at most $\max\{1\cdot4~,~2\cdot3\}=6$ such paths. Thus}
&6|S|\ge\sum_v \deg_{C}(v)\bigl(5-\deg_{C}(v)\bigr)
 \ge\frac{9|X||\oX|}{n}.\notag\\
\shortintertext{The cut-graph has maximum degree at most five and no four-cycle. Applying \cref{lemma:cut-degree-matching} with $\Delta=5$ (taking the cut-graph cannot increase the maximum degree of the graph) gives}
&\rho_G(X)\ge\frac{|S|}{5}
 \ge\frac{3|X||\oX|}{10n}.\notag
\end{align}
\endgroup
\else
\begin{equation}\label{eq:ramanujan-five-spectral}
 \sum_{v\in V(G)} \left(|N_G(v)\cap X|-\frac{5|X|}{n}\right)^2 \le16\frac{|X||\oX|}{n}.
\end{equation}
Here $|N_G(v)\cap X|$ equals $\deg_{C}(v)$ when $v\in\oX$, and $5-\deg_{C}(v)$ when $v\in X$.

Each vertex of $X$ belongs to $N_G(v)\cap X$ for exactly five vertices $v$, namely its five neighbours. Thus
\begin{equation}\label{identity-X}    
 \sum_{v\in V(G)}|N_G(v)\cap X|=5|X|
\end{equation}

We now expand each square and sum its three terms. The constant term occurs once for each of the $n$ vertices, giving
\begin{align}\label{expansion-square}
 &\sum_v\left(|N_G(v)\cap X|-\frac{5|X|}{n}\right)^2\notag\\
 \quad=&\sum_v|N_G(v)\cap X|^2-\frac{10|X|}{n}\sum_v|N_G(v)\cap X|+n\left(\frac{5|X|}{n}\right)^2\notag\\
 \quad=&\sum_v|N_G(v)\cap X|^2-\frac{50|X|^2}{n}+\frac{25|X|^2}{n}\notag\\
 \quad=&\sum_v|N_G(v)\cap X|^2-\frac{25|X|^2}{n}
\end{align}

Rearranging Eq.~\eqref{expansion-square}, we get
\begin{equation}\label{expansion-square2}
 \sum_v|N_G(v)\cap X|^2 =\frac{25|X|^2}{n} +\sum_v\left(|N_G(v)\cap X|-\frac{5|X|}{n}\right)^2
\end{equation}

Summing over middle vertices counts each such path once. Their total number therefore satisfies
\begin{align}\label{path-count} 
 \sum_v \deg_{C}(v)\bigl(5-\deg_{C}(v)\bigr) &=\sum_v |N_G(v)\cap X|\bigl(5-|N_G(v)\cap X|\bigr)\notag\\
 &=5\sum_v|N_G(v)\cap X|-\sum_v|N_G(v)\cap X|^2\notag\\
 &=25|X|-\sum_v|N_G(v)\cap X|^2
\end{align}
where the last equality uses Eq.~\eqref{identity-X}.

Substituting Eq.~\eqref{expansion-square2} into the path count Eq.~\eqref{path-count}, and using
$n-|X|=|\oX|$, we obtain
\begin{align}
 \sum_v \deg_{C}(v)\bigl(5-\deg_{C}(v)\bigr)
 &=25|X|-\frac{25|X|^2}{n}
   -\sum_v\left(|N_G(v)\cap X|-\frac{5|X|}{n}\right)^2\notag\\
 &=\frac{25|X||\oX|}{n}
   -\sum_v\left(|N_G(v)\cap X|-\frac{5|X|}{n}\right)^2\notag\\
 &\ge\frac{25|X||\oX|}{n}-\frac{16|X||\oX|}{n}=\frac{9|X||\oX|}{n}.
 \label{eq:ramanujan-five-pairs}
\end{align}
The last inequality is Eq.~\eqref{eq:ramanujan-five-spectral}.


Every such path has its middle vertex in $S$ by construction. There are only two cases for the number of neighbours a vertex of $S$ on each cut side: either one and four, or two and three. So it is the middle vertex of at most $\max\{1\cdot4~,~2\cdot3\}=6$ such paths. Thus
\[
 6|S|\ge\sum_v \deg_{C}(v)\bigl(5-\deg_{C}(v)\bigr)
 \ge\frac{9|X||\oX|}{n}.
\]
The cut-graph has maximum degree at most five and no four-cycle.
Applying \cref{lemma:cut-degree-matching} with $\Delta=5$ (taking the cut-graph cannot increase the maximum degree of the graph) gives
\[
 \rho_G(X)\ge\frac{|S|}{5}
 \ge\frac{3|X||\oX|}{10n}.
\]\fi\end{proof}
\fi

\begin{proof}[\cref{thm:ram}]
Morgenstern's construction~\cite[Thm.~5.13]{Mor94}, with prime-power parameter $q=4$, gives an infinite family of connected nonbipartite $q+1=5$-regular Ramanujan graphs whose girth tends to infinity as a $\Omega(\log n)$. 
After omitting finitely many members, none contains a four-cycle. Fix any retained $n$-vertex graph $G$. For any balanced cut $(X,\oX)$, we have $|X||\oX|\ge\frac{2n^2}{9}$, which by \cref{prop:ramanujan-five} gives $\rho_G(X)\ge\frac{3}{10n}\cdot\frac{2n^2}{9}=\frac n{15}.$
By \cref{lemma:rw}, any such retained graph, deterministically constructed in time polynomial in $n$, has rank-width at least $n/15$.\end{proof}
For comparison, for random graphs $G(n,p)$ with fixed $0<p<1$, the rank-width is within an additive constant of its maximum possible value $\lceil n/3\rceil$, w.h.p. \cite{LLO12}.




\section{Cartesian products}
\ifbriefprelim
The \emph{Cartesian product} $G \square H$ of $G,H$ is the graph with vertex set $V(G) \times V(H)$, and edge set $\set{(u, x)(v, y) : (u,v)\in E(G) \text{ and } x = y\text{, or }(x,y)\in E(H)\text{ and } u = v }$. 

Given graphs $G_1,\dots,G_k$, we write their Cartesian product as $\bigsquare_{i=1}^k G_i$ or simply 
${G}^{\square k}$ if $G_1 = \cdots = G_k = G$.
\fi




Instead of greedily constructing a large matching, we will recursively collect many candidates by stripping the product graph coordinate by coordinate and combine them on the way back while certifying acyclicity for the final resulting matching. In the end, at most one edge of each Cartesian copy of a graph factor will be selected, although it is already sufficient to state the following.

\thmcartesian*\label{page:cartesian}

Note that Eq.~\eqref{eq:c-cut} holds for every cut, not necessarily balanced.

\iffullversion
\begin{proof}
Fix a cut $(X,\oX)$. We construct a $(X,\oX)$-matching $M$ recursively, then retain at least half of its selected edges as a directed acyclic matching.

\paragraph{Construction.}
Fix a root in each $G_i$, and let $A_i,B_i$ be its vertices at odd and even distance from the root, respectively. Every nontrivial cut of $G_i$ has a crossing edge joining $A_i$ to $B_i$\,: a shortest path from the root to the other cut side crosses between consecutive distance levels.
Fix an arbitrary order on the edges joining $A_i$ to $B_i$ in each factor, and choose the first crossing one whenever a choice is needed. Use these same classes and orders in every copy.

In recursive calls, the inherited partition may have an empty side, in that case return the empty matching.

For $k=1$ (one factor remaining) and both sides nonempty, return one crossing edge joining $A_1$ to $B_1$ according to the arbitrary edge order chosen.

For $k\ge2$, we will group the $k-1$ first coordinates together and write a vertex as
\[
(u,v)\text{, where }u\in V(G_1\square\cdots\square G_{k-1})\text{ and }v\in V(G_k) 
\]
The sections $S_v=V(G_1\square\cdots\square G_{k-1})\times\{v\}$ are $n_k$ disjoint copies of the smaller product, each of order $n/n_k$. The coordinate copies $C_u=\{u\}\times V(G_k)$ are $n/n_k$ disjoint copies of $G_k$, each of order $n_k$. We call a coordinate copy $C_u$ \emph{mixed} if it meets both cut sides.

\noindent Form two \emph{candidate} matchings:
\begin{itemize}
 \item $M_{\mathrm{between}}$ contains the first crossing edge joining $A_k$ to $B_k$ in every mixed $C_u$ (if a copy $C_u$ does not meet both cut sides, it contains no cut-edge and cannot contribute to the matching);
 \item $M_{\mathrm{inside}}$ is the union of the matchings $M_v$ returned by the recursive call of this algorithm within each section $S_v$, with inherited cut sides $X\cap S_v$ and $\oX\cap S_v$.
\end{itemize}

The coordinate copies are pairwise disjoint, as are the sections, so each candidate forms a matching.
Return the larger candidate, choosing $M_{\mathrm{inside}}$ in a tie. Consequently,
\begin{equation}\label{eq:ce-candidates}
 |M|\ge |M_{\mathrm{between}}|,\quad |M|\ge |M_{\mathrm{inside}}|=\sum_v|M_v|.
\end{equation}
Figure~\ref{fig:recursion} illustrates the recursive calls.

\begin{figure}[ht]
 \centering
 \includegraphics[width=.85\linewidth]{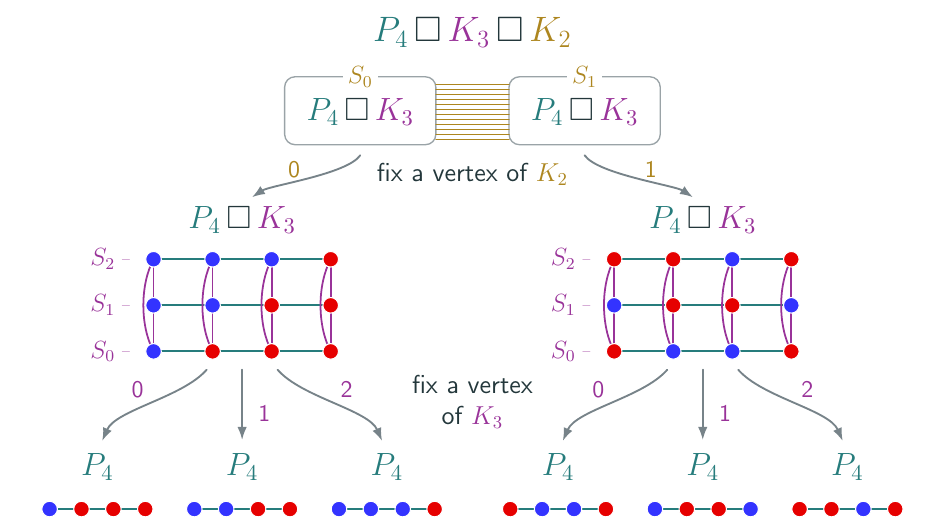}
 \caption{Recursive sections for the Cartesian product of  $G_1=P_4$, $G_2=K_3$ and $G_3=K_2$. Fixing the last coordinate gives two sections isomorphic to $G_1\square G_2$; fixing the $G_2$ coordinate next gives six sections, each corresponding to a $G_1$, that inherit the fixed cut and to which the base case is applied. Blue and red indicate the inherited cut sides $X$ and $\oX$. The top boxes schematically represent the sections obtained by fixing a bit for the last coordinate; on the first step a candidate matching will be formed from the gold edges shown between them. Arrows show recursive calls to form the second candidate from edges within both sections; the two candidates are compared upon return.}
 \label{fig:recursion}
\end{figure}

\paragraph{Acyclicity.}
We prove by induction on $k$ that, for every cut $(X,\oX)$, both orientation classes of the matching returned by the construction are directed acyclic.

Every edge of a Cartesian product changes exactly one coordinate. Orient each selected edge from $A_i$ to $B_i$ in the coordinate it changes. The \emph{height} of a product vertex is the number of coordinates $i$ whose entries lie in $B_i$. Every selected edge raises height by one and is called ascending, whereas every edge of $G$ changes height by at most one. The two orientation classes consist of the selected edges ascending from $X$ to $\oX$, and those ascending from $\oX$ to $X$.

For $k=1$, the returned matching has at most one edge, so the acyclicity claim holds.
For $k\ge2$, assume that the claim holds for every cut of a product of $k-1$ factors. Let $\widehat M$ be either orientation class. By exchanging $X$ and $\oX$ if necessary, assume that its edges ascend from $X$ to $\oX$. For distinct edges $e_0=(u_0,v_0)$ and $e_1=(u_1,v_1)$ of $\widehat M$, with $u_0,u_1\in X$ and $v_0,v_1\in\oX$, recall that by definition $A_{\widehat M,X}$ has an arc $e_0\to e_1$ exactly when $(v_0,u_1)$ is an edge of $G$.

Suppose first that $M=M_{\mathrm{between}}$. Recall that each copy $C_u$ fixes the first $k-1$ coordinates and lets only coordinate $k$ vary. Since the construction selects at most one edge in each copy, $e_0$ and $e_1$ lie in different copies. These copies have different fixed $(k-1)$-tuples, so $v_0$ and $u_1$ differ in at least one of the first $k-1$ coordinates.
They also differ in coordinate $k$\,: both selected edges are oriented from $A_k$ to $B_k$, so the $k$-th coordinate of $v_0$ lies in $B_k$, whereas that of $u_1$ lies in $A_k$. Thus they differ in at least two coordinates and cannot be adjacent in $G$.
Hence $A_{\widehat M,X}$ has no arcs. In particular, $\widehat M$ is an induced matching in the cut-graph $G[\partial X]$ and is directed acyclic.

Suppose now that $M=M_{\mathrm{inside}}$ and that $A_{\widehat M,X}$ contains a directed cycle. By the definition of the edge-adjacency digraph, the cycle describes the following closed walk in the original graph: traverse a \emph{selected} ascending edge, then a \emph{returning} edge to the lower endpoint of the next selected edge, and repeat. The walk contains equally many selected and returning edges. Each selected edge raises height by one, while each returning edge lowers it by at most one. Since the walk is closed, every returning edge must lower height by exactly one.

A selected $M_{\mathrm{inside}}$ edge leaves coordinate $k$ fixed. A returning edge that changes the last coordinate must move it from $B_k$ to $A_k$, since it must lower height. To close the walk and return to the starting vertex coordinate, a later step would have to move the last coordinate back from $A_k$ to $B_k$. This is impossible: selected edges leave coordinate $k$ fixed, and a returning edge moving it from $A_k$ to $B_k$ would raise height. Hence no step of the closed walk can change coordinate $k$.

Every vertex of the walk therefore lies in the same section $S_v$. Within this section, $M_{\mathrm{inside}}$ consists precisely of the edges returned by the recursive call, namely $M_v$. The selected edges of the walk therefore belong to $M_v$ and still ascend from $X\cap S_v$ to $\oX\cap S_v$. Every returning edge also lies within $S_v$, so it still gives an arc from one selected edge to the next in the recursive edge-adjacency digraph. Thus the same directed cycle occurs in one orientation class of $M_v$. But $S_v$ is a product of $k-1$ factors, and the induction hypothesis says that both orientation classes returned by this recursive call are directed acyclic, reaching a contradiction.

Thus both orientation classes of $M$ are directed acyclic.
Choosing $\widehat M$ to be the larger class,
\cref{lemma:jelinek} gives
\begin{equation}\label{eq:matching-rank}
 \rho_G(X)\ge |\widehat M|\ge\frac{|M|}{2}.
\end{equation}

\paragraph{Counting.}
Let $P=|X||\oX|$ count pairs of vertices on opposite cut sides, whether adjacent in $G$ or non-edges. For $k\ge2$, in each section $S_v$ put
$x_v=|X\cap S_v|$ and $p_v=x_v\left(\frac{n}{n_k}-x_v\right)$. Thus $p_v$ counts pairs with one vertex in $X\cap S_v$ and the other in $\oX\cap S_v$. 

Now, between two distinct sections $S_v$ and $S_w$, the number of opposite-side pairs is
\begin{equation}\label{eq:ce-two-sections}
 \begin{aligned}
 &x_v\left(\frac n{n_k}-x_w\right) +x_w\left(\frac n{n_k}-x_v\right)\\
 &\qquad=\frac{n}{n_k}(x_v+x_w)-2x_vx_w\\
 &\qquad=\left[\frac{n}{n_k}(x_v+x_w)-x_v^2-x_w^2\right]+\left[x_v^2+x_w^2 -2x_vx_w\right]\\
 &\qquad=p_v+p_w+(x_v-x_w)^2
 \end{aligned}
\end{equation}
To get all types of pairs in $P$, we have to sum over all unordered pairs $\{v,w\}$ of distinct section
labels, and add the pairs lying within individual sections:
\begin{equation}
 \begin{aligned}
 P &=\sum_v p_v+\sum_{\{v,w\}}\bigl(p_v+p_w+(x_v-x_w)^2\bigr)\\
 &=\sum_v p_v+(n_k-1)\sum_v p_v+\sum_{\{v,w\}}(x_v-x_w)^2\\
 &=n_k\sum_v p_v+\sum_{\{v,w\}}(x_v-x_w)^2
 \end{aligned}
\end{equation}
Indeed, each $p_v$ appears exactly $n_k-1$ times in the sum over unordered pairs of distinct sections, once for each other section. Together with its contribution from within $S_v$, its total coefficient is therefore $n_k$. We write $P=A+B$\,:
\begin{equation}\label{eq:ce-pairs}
 P=\underbrace{n_k\sum_vp_v}_{A}
   +\underbrace{\sum_{\{v,w\}}(x_v-x_w)^2}_{B}.
\end{equation}
We will control $A$ by recursion in the sections and $B$ by the mixed
coordinate copies. We now prove the following hypothesis by induction, to relate $P$ to $|M|$\,: 
\begin{equation}\label{eq:ce-induction}
 8P^2\le n^2|M|^2\sum_{i=1}^k n_i^2
 \end{equation}
We induct on $k$. Having either side of the cut empty gives $P=0$. For $k=1$ and both sides nonempty, $n=n_1$ and our algorithm gives $|M|=1$. Since $|X|+|\oX|=n$, we have $P\le \frac{n^2}{4}$ and $8P^2\le \frac{n^4}{2}\le n^2n_1^2|M|^2$ indeed holds.

Suppose $k\ge2$ and assume Eq.~\eqref{eq:ce-induction} holds for every cut of $G_1\square\cdots\square G_{k-1}$. We begin by bounding the first term $A$ of Eq.~\eqref{eq:ce-pairs} using recursion within the sections. Each $S_v$ induces a product of $G_1,\ldots,G_{k-1}$ of order $\frac{n}{n_k}$, with pair count $p_v$ and returned matching $M_v$. The induction hypothesis applied to $(X\cap S_v,\oX\cap S_v)$ then gives:
\begin{equation}\label{eq:ce-section-hypothesis}
 8p_v^2\le\left(\frac n{n_k}\right)^2|M_v|^2
             \sum_{i=1}^{k-1}n_i^2
\end{equation}
Taking square roots gives
\begin{equation}\label{eq:ce-section-induction}
 \sqrt8\,p_v\le\frac n{n_k}
       \sqrt{\sum_{i=1}^{k-1}n_i^2}\,|M_v|
\end{equation}
The factor $\frac n{n_k}\sqrt{\sum_{i=1}^{k-1}n_i^2}$ does not depend on the section label $v$. Adding the section-dependent inequality Eq.~\eqref{eq:ce-section-induction} over $v\in V(G_k)$ therefore gives
\begin{equation}\label{eq:ce-section-sum}
 \sqrt8\sum_vp_v
 \le\frac n{n_k}\sqrt{\sum_{i=1}^{k-1}n_i^2}\,
       \underbrace{\sum_v|M_v|}_{|M_{\mathrm{inside}}|\,\le\,|M|}.
\end{equation}
Multiply by $n_k$ to recover $A=n_k\sum_vp_v$, and use the
candidate comparison Eq.~\eqref{eq:ce-candidates}:
\begin{equation}\label{eq:ce-recover-A}
 \sqrt8\,A\le n\sqrt{\sum_{i=1}^{k-1}n_i^2}\,|M|.
\end{equation}
Squaring both sides gives
\begin{equation}\label{eq:ce-inside}
 8A^2\le n^2|M|^2\sum_{i=1}^{k-1}n_i^2.
\end{equation}
This accounts for the first $k-1$ factors in the induction hypothesis inequality Eq.~\eqref{eq:ce-induction}. To include the last factor $G_k$, it remains to prove
\begin{equation}\label{eq:ce-missing}
 8P^2-8A^2\le n^2n_k^2|M|^2.
\end{equation}

To bound the differences, fix distinct sections $S_v,S_w$ and suppose $x_v\ge x_w$. Pair each vertex $(u,v)\in X\cap S_v$ with its corresponding vertex $(u,w)\in S_w$, as in \cref{fig:ce-surplus}. At most $x_w$ partners belong to $X$ by definition. Thus at least $x_v-x_w$ belong to $\oX$, each in distinct mixed copies $C_u$. Each such copy contributes one edge to $M_{\mathrm{between}}$ by construction. Exchanging $v,w$ if necessary, this gives using Eq.~\eqref{eq:ce-candidates}
\begin{equation}\label{eq:ce-surplus}
 |x_v-x_w|\le |M_{\mathrm{between}}|\le |M|.
\end{equation}

\begin{figure}[htbp]
 \centering
 \includegraphics[width=.90\linewidth]{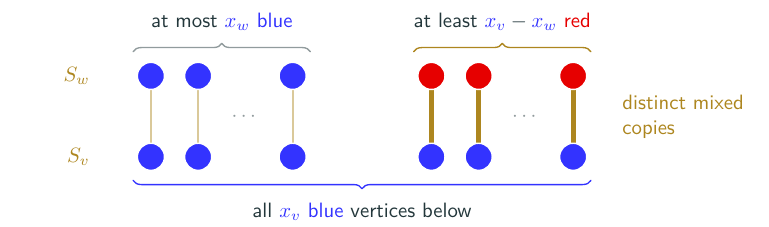}
 \caption{Minimum number of opposite-side pairs between $S_v$ and $S_w$, for $x_v\ge x_w$. Blue means $X$ and red means $\oX$. Only the $x_v$ blue vertices of $S_v$ and their corresponding pair partners in $S_w$ are shown. Segments indicate pairing, not necessarily graph edges.}
 \label{fig:ce-surplus}
\end{figure}

\begin{samepage}
To bound the second term $B$ of Eq.~\eqref{eq:ce-pairs}, apply the bound Eq.~\eqref{eq:ce-surplus} to one factor in each square:
\begin{equation}\label{eq:ce-linear}
 B=\sum_{\{v,w\}}(x_v-x_w)^2\le |M_{\mathrm{between}}|\sum_{\{v,w\}}|x_v-x_w|.
\end{equation}
Let $Q$ count all pairs with one vertex in $X$, the other in $\oX$, and both vertices in the same coordinate copy $C_u$. 
Since all copies are pairwise disjoint, we simply get $Q=\sum_u |X\cap C_u||\oX\cap C_u|\,.$

Now group the pairs counted by $Q$ according to the unordered pair of sections containing their vertices. Each $C_u$ contains only one vertex of each section by structure of the Cartesian product, so the two sections of a pair are distinct. 
For fixed sections $S_v,S_w$, collect all pairs $(u,v),(u,w)$ on opposite cut sides, with $u$ ranging over all coordinate copies. The surplus argument used in Eq.~\eqref{eq:ce-surplus} gives at least $|x_v-x_w|$ pairs in this group. The sizes of these disjoint groups sum to $Q$. Adding their lower bounds gives
\begin{equation}\label{eq:ce-surplus-pairs}
 \sum_{\{v,w\}}|x_v-x_w|\le Q.
\end{equation}

In each mixed $C_u$, the two cut-side sizes sum to $n_k$, so again their product is at most $n_k^2/4$. There are exactly $|M_{\mathrm{between}}|$ mixed copies by definition, so we can also upper-bound $Q$:
\begin{equation}\label{eq:ce-copy-pairs}
Q\le \frac{n_k^2}{4}|M_{\mathrm{between}}|.
\end{equation}
Insert Eq.~\eqref{eq:ce-surplus-pairs} and Eq.~\eqref{eq:ce-copy-pairs} into Eq.~\eqref{eq:ce-linear}, and use the candidate
comparison Eq.~\eqref{eq:ce-candidates}:
\begin{equation}\label{eq:ce-between}
 B\le |M_{\mathrm{between}}|Q
 \le\frac{n_k^2|M_{\mathrm{between}}|^2}{4}
 \le\frac{n_k^2|M|^2}{4}.
\end{equation}
\end{samepage}

We now bound the remaining difference $8P^2-8A^2$ in Eq.~\eqref{eq:ce-missing}. Since $P-A=B$ by the pair-counting identity Eq.~\eqref{eq:ce-pairs}:
\begin{equation}\label{eq:ce-difference}
 P^2-A^2=(P-A)(P+A)=B(P+A).
\end{equation}
It remains to bound $P+A$. The identity $P=A+B$ gives $0\le A\le P$. Also recall that $P\le \frac{n^2}{4}$, therefore
\begin{equation}\label{eq:ce-remaining-factor}
 P+A\le 2P\le\frac{n^2}{2}.
\end{equation}
Insert the bounds Eq.~\eqref{eq:ce-between} and Eq.~\eqref{eq:ce-remaining-factor} into the difference-of-squares identity Eq.~\eqref{eq:ce-difference}, and multiply by $8$:
\begin{equation}\label{eq:ce-remaining-contribution}
 \begin{aligned}
 8P^2-8A^2
 &=8B(P+A)\\
 &\le 8\cdot\frac{n_k^2|M|^2}{4}\cdot\frac{n^2}{2}
  =n^2n_k^2|M|^2.
 \end{aligned}
\end{equation}
This is the remaining contribution in Eq.~\eqref{eq:ce-missing}. Adding it to the bound we obtained for the first term $A$ using the first $k-1$ factors in Eq.~\eqref{eq:ce-inside}, we get
\begin{equation}\label{eq:ce-finish}
 \begin{aligned}
 8P^2
 &=8A^2+8B(P+A)\\
 &\le n^2|M|^2\sum_{i=1}^{k-1}n_i^2+n^2n_k^2|M|^2\\
 &=n^2|M|^2\sum_{i=1}^k n_i^2.
 \end{aligned}
\end{equation}
This proves our induction hypothesis in Eq.~\eqref{eq:ce-induction}. Taking square roots and rearranging yields
\begin{equation}\label{eq:ce-rank-conversion}
 |M|\ge\frac{\sqrt8\,P}{n\sqrt{\sum_{i=1}^k n_i^2}}.
\end{equation}
By the acyclicity of either direction choice, every cut $(X,\oX)$ therefore satisfies
\begin{equation}\label{eq:ce-cut-conclusion}
 \rho_G(X)\ge\frac{|M|}{2}
 \ge\frac{\sqrt2\,|X||\oX|}{n\sqrt{\sum_{i=1}^k n_i^2}}.
\end{equation}
For a balanced cut, $|X||\oX|\ge\frac{2n^2}{9}$, thus \cref{lemma:rw} gives
\begin{equation}\label{eq:ce-rw-conclusion}
 \rw(G)\ge
 \frac{\sqrt2\cdot(2n^2/9)}{n\sqrt{\sum_{i=1}^k n_i^2}}
 =\frac{2\sqrt2}{9}\cdot\frac{n}{\sqrt{\sum_{i=1}^k n_i^2}}.
\end{equation}\end{proof}
\else
\begin{proof}[Sketch]
Fix a cut $(X,\oX)$. We construct a $(X,\oX)$-matching $M$ recursively, then retain at least half of its selected edges as a directed acyclic matching.

\paragraph{Construction.}
Fix a root in each $G_i$, and let $A_i,B_i$ be its vertices at odd and even distance from the root, respectively. Every nontrivial cut of $G_i$ has a crossing edge joining $A_i$ to $B_i$\,: a shortest path from the root to the other cut side crosses between consecutive distance levels.
Fix an arbitrary order on the edges joining $A_i$ to $B_i$ in each factor, and choose the first crossing one whenever a choice is needed. Use these same classes and orders in every copy.

In recursive calls the inherited partition may have an empty side. In that case return the empty matching.

For $k=1$ (one factor remaining) and both sides nonempty, return one crossing edge joining $A_1$ to $B_1$ according to the arbitrary edge order chosen.

For $k\ge2$, we will group the $k-1$ first coordinates together and write a vertex as $(u,v)$, where $u\in V(G_1\square\cdots\square G_{k-1})$ and $v\in V(G_k)$.
The sections $S_v=V(G_1\square\cdots\square G_{k-1})\times\{v\}$ are $n_k$ disjoint copies of the smaller product, each of order $n/n_k$. The coordinate copies $C_u=\{u\}\times V(G_k)$ are $n/n_k$ disjoint copies of $G_k$, each of order $n_k$. We call a coordinate copy $C_u$ \emph{mixed} if it meets both cut sides.

\noindent Form two \emph{candidate} matchings:
\begin{itemize}
 \item $M_{\mathrm{between}}$ contains the first crossing edge joining $A_k$ to $B_k$ in every mixed $C_u$ (if a copy $C_u$ does not meet both cut sides, it contains no cut-edge and cannot contribute to the matching);
 \item $M_{\mathrm{inside}}$ is the union of the matchings $M_v$ returned by the recursive call of this algorithm within each section $S_v$, with inherited cut sides $X\cap S_v$ and $\oX\cap S_v$.
\end{itemize}

The coordinate copies are pairwise disjoint, as are the sections, so each candidate forms a matching.
Return the larger candidate, choosing $M_{\mathrm{inside}}$ in a tie. Consequently
\ifbriefprelim 
\ifcompactlayout,\else
(Figure~\ref{fig:recursion} in Appendix~\ref{app:proof-cartesian} illustrates the recursive calls),
\fi
\fi

\begin{equation}\label{eq:ce-candidates}
 |M| ~\ge~ |M_{\mathrm{between}}|\quad\mbox{and}\quad |M| ~\ge~ |M_{\mathrm{inside}}| = \sum_v|M_v| ~.
\end{equation}

\ifbriefprelim\else
\begin{figure}[ht]
 \centering
 \includegraphics[width=.85\linewidth]{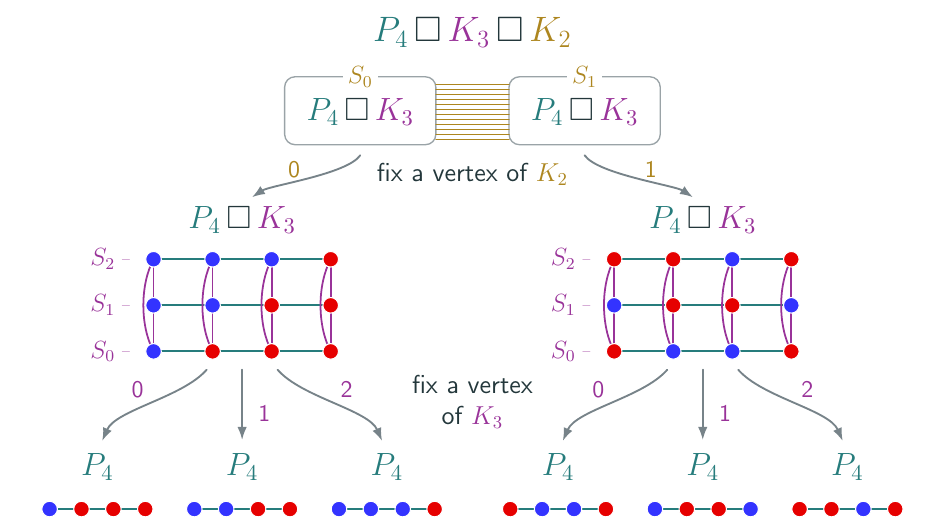}
 \caption{Recursive sections for the Cartesian product of  $G_1=P_4$, $G_2=K_3$ and $G_3=K_2$. Fixing the last coordinate gives two sections isomorphic to $G_1\square G_2$; fixing the $G_2$ coordinate next gives six sections, each corresponding to a $G_1$, that inherit the fixed cut and to which the base case is applied. Blue and red indicate the inherited cut sides $X$ and $\oX$. The top boxes schematically represent the sections obtained by fixing a bit for the last coordinate; on the first step a candidate matching will be formed from the gold edges shown between them. Arrows show recursive calls to form the second candidate from edges within both sections; the two candidates are compared upon return.}
 \label{fig:recursion}
\end{figure}

\fi
\paragraph{Acyclicity.}
Orient each selected edge from $A_i$ to $B_i$ in the coordinate it changes,
and separate the two orientation classes across the cut. The height of
a product vertex counts its coordinates in $B_i$: selected edges raise
height by one, while every graph edge changes it by at most one.
Each orientation class of $M_{\mathrm{between}}$ is induced in the
cut-graph; for $M_{\mathrm{inside}}$, after exchanging cut sides if
necessary, a directed cycle would alternate
ascending selected edges with returning edges that must all lower
height. Selected inside edges leave coordinate $k$ fixed, while returning edges can change it only from $B_k$ to $A_k$. A closed walk therefore cannot change coordinate $k$, so the cycle stays in one section.
Induction there excludes the cycle, so the larger orientation class
gives $\rho_G(X)\ge |M|/2$; the full acyclicity argument is in
\ifcompactlayout
the full version \cite{CGLBM26}.
\else
Appendix~\ref{app:proof-cartesian}.
\fi

\paragraph{Counting.}
Let $P=|X||\oX|$ count pairs of vertices on opposite cut sides, without
requiring adjacency. The induction goal is
\begin{equation}\label{eq:ce-verification-induction}
 8P^2 ~\le~ n^2|M|^2\sum_{i=1}^k n_i^2 ~.
\end{equation}
We break down the quantity $P$ over two sums each regarding some cut-vertex sets restricted to some section:
\begin{equation}
 P
 ~=~ n_k\sum_v |X\cap S_v|\,|\oX\cap S_v|
   ~+~ \sum_{\{v,w\}}
      \bigl(|X\cap S_v|-|X\cap S_w|\bigr)^2,
\end{equation}
where the second sum runs over unordered pairs of distinct sections.
This identity connects the pair count to our two candidates. The first term is $n_k$ times the number of opposite-side pairs within sections, which recursion controls through $M_{\mathrm{inside}}$.
For the second term, differences in the numbers of vertices of $X$ in two sections force mixed coordinate copies, each contributing one edge to $M_{\mathrm{between}}$. 
Comparing the two candidates bounds both terms and yields
\begin{equation}\label{eq:ce-verification-matching}
 |M| ~\ge~ \frac{\sqrt8\,|X||\oX|}{n\sqrt{\sum_{i=1}^k n_i^2}} ~.
\end{equation}
Together with $\rho_G(X)\ge |M|/2$, this proves the every-cut bound,
and $|X||\oX|\ge2n^2/9$ gives the rank-width conclusion by
\cref{lemma:rw}; the complete pair count is in
\ifcompactlayout
the full version \cite{CGLBM26}.
\else
Appendix~\ref{app:proof-cartesian}.
\fi
\end{proof}
%
\fi

In particular, \cref{thm:cartesian} gives the following direct corollary for Cartesian powers of a graph.
\begin{corollary}\label{cor:cartesian}
    Let $G$ be a connected graph with $n\ge 2$ vertices. For every $k\ge1$,
 $$\rw(G^{\square k})\ge\frac{2\sqrt{2}}{9}\cdot\frac{n^{k-1}}{\sqrt{k}}\,.$$
\end{corollary}

\noindent In particular, the $k$-dimensional hypercube $Q_k=K_2^{\square k}$ satisfies the lower bound

\begin{equation}\label{hypercube}
    \rw(Q_k)\ge\frac{\sqrt{2}}{9}\cdot\frac{2^k}{\sqrt{k}}\,.
\end{equation}


See \cref{sec:other} for upper bound discussions.
\iffullversion
Note that if $G$ is a not a connected graph and is not edgeless, let $s\ge2$ be the order of its largest connected component, then the following lower bound holds
$\rw(G^{\square k})=\Omega\pare{\frac{s^{k-1}}{\sqrt{k}}}$.
The product of the largest component witnesses the rank-width.
\fi
Alternatively, if we take factors of bounded order, we get the following corollary.

\begin{corollary}\label{cor:bounded-factor}
Let $C>1$ be a fixed integer. For $k\ge1$ and every $i\in\range{1}{k}$, let $G_i$ be a connected graph with $2\le n_i\le C$ vertices. Let $G=\bigsquare_{i=1}^kG_i$ be their Cartesian product with $n=\prod_{i=1}^k n_i$ vertices. Then,
 $\rw(G)\ge\frac{2\sqrt{2}}{9C}\cdot\frac{n}{\sqrt{k}}$.
\end{corollary}

\begin{remark}\label{dekk-answer}
Diskin, Erde, Kang and Krivelevich \cite[Question~7.3(1)]{diskin2024isoperimetric} ask whether, for a fixed integer $C>1$, every product $G=\square_{i=1}^kG_i$ of connected $d_i$-regular graphs with $2\le |V(G_i)|\le C$ is s.t.:
$|N_G(X)|=\Omega\!\left(\frac{|X|}{\sqrt d}\right)~ \forall X\subseteq V(G), |X|\le \frac{n}{2},$
where $n=|V(G)|$, $d=\sum_i d_i$, and $N_G(X)$ is the set of vertices in $\overline{X}$ with a neighbour in $X$, often called the \emph{outer boundary} of $X$, closely related to the study of vertex expansion. The matching estimate Eq.~\iffullversion\eqref{eq:ce-rank-conversion}\else\eqref{eq:ce-verification-matching}\fi{} in the proof of \cref{thm:cartesian} answers: $|N_G(X)| \ge |M|
 \ge \frac{\sqrt2}{C}\cdot\frac{|X|}{\sqrt{k}} ~\,\pare{\sum_i 1=k\le d=\sum_i d_i}$. 
 \ifcompactlayout
 See \cite{CGLBM26} for details.
\else
 See Appendix~\ref{app:question} for more details.
\fi
\end{remark}
\section{Quantum States Applications and Discussion}
\label{sec:discussion_trunc}

\subsection{Graph state preparation}\label{sec:gs}

A graph state $\ket{G}$ on $n$ qubits is a stabilizer state represented by an undirected graph $G$. To prepare $\ket{G}$, one only has to prepare the state $\ket{+}^{\otimes n}=\ket{\overline{K_n}}$ (using a layer of Hadamard gates) and for each edge $\{u,v\}\in E(G)$ apply a controlled-Z (CZ) operation on the corresponding qubits $u$ and $v$, assuming all-to-all connectivity:
\vspace{-2mm}
$$\ket{G}=\left[\prod\nolimits_{\{u,v\}\in E(G)}\mathrm{CZ}^{\{u,v\}}\right]\ket{+}^{\otimes n}.$$
This method allows us to prepare $\ket{G}$ using exactly $\lvert E(G)\rvert$ two-qubit gates, all pairwise commuting. Moreover, by Vizing's theorem, the edges of $G$ can be partitioned into at most $\Delta(G)+1$ matchings. Each matching corresponds to a set of two-qubit gates that can be executed in parallel, so $\ket{G}$ can be prepared with a circuit $C$ of two-qubit-gate depth at most $\Delta(G)+1$. 

Recall the following well-known proposition on Schmidt rank, which provides one way to measure the entanglement shared by the two subsystems defined by a bipartition. For a graph state, this entanglement can be viewed as a collection of independent Bell pairs, each consisting of two maximally entangled qubits on opposite sides. 
To globalize from a single bipartition to the whole state, we consider the cuts of a decomposition tree: taking the largest such number for each tree, then minimizing over all trees, gives the entanglement width.
\begin{proposition}[\cite{hein2004multiparty,van2006universal}]
Given a graph state $\ket{G}$ and a bipartition $\pare{X,\oX}$, the Schmidt-rank of $\ket{G}$ across $\pare{X,\oX}$ is $2^{\rho_G(X)}$.
\end{proposition}

This implies that \emph{every subcubic tree decomposition of the qubits of $\ket{G}$} contains an edge cut with Schmidt-rank at least $2^{\rw(G)}$.
Then, generating explicit graphs with near-linear rank-width gives graph states of near-linear entanglement width, which is a critical stepping stone for achieving practical quantum advantage experiments. Indeed, these constructions are natural candidates for benchmarking \emph{quantum} circuit instances with certified worst-case entanglement, complementing the average-case guarantees of random regular graph states. On the other hand, rank-width lower-bounds the required bond dimension of a tree tensor network \cite{o2019parameterization}: if $\rw(G)\geq r$, then any subcubic tree tensor network, with one qubit per leaf, representing $\ket{G}$ exactly requires bond dimension at least $2^r$ across some edge. Therefore, these deterministic large rank-width families are adversarial benchmark instances for tensor-network/MPS-based \emph{classical} simulators, arguably the most common simulation methods to date.

Crucially, our expansion-based framework (via \cref{thm:algo}) provides excellent candidates to build asymptotically near-linear entanglement in bounded or logarithmic CZ depth under all-to-all connectivity through \emph{deterministic families} of graphs of large rank-width. 
Indeed, random regular graph states of fixed degree have been shown to generate linear entanglement \emph{with high probability} \cite{LLO12}. These states can be prepared by a constant depth circuit under all-to-all connectivity and can be used in practical quantum advantage experiments \cite{martiel2026sampling}.
Note that for a random 3-regular graph $G$ by \cite[Theorem~4.4]{DEK24}, we have $h(G) = \alpha_3 \ge 1/4.95$. Thus, by Eq.~\eqref{eq:h(G)}, 
prior to our work, such graphs satisfy w.h.p.:
\begin{equation}\label{prior-3-reg}
    \rw(G) ~\ge~ \frac{h(G)\cdot |V(G)|/3}{\Delta(G)^2} = \frac{\alpha_3\cdot n/3}{3^2} > \frac{n}{4.95\cdot 3\cdot 9} > \frac{n}{134} ~.
\end{equation}
This work shows that $d$-regular (near-)Ramanujan graphs, some of which have explicit constructions \cite{lubotzky1988ramanujan,mohanty2020explicit,alon2020explicitexpandersdegreesize} for each fixed $d\ge3$, have rank-width $\Theta(n)$. The corresponding graph states can be prepared in constant depth under all-to-all connectivity and have certified linear entanglement, i.e., deterministically (see \cref{tableres} for a summary of our results). 
Therefore, these states have the same asymptotic entanglement-width guarantee as random regular graph states, while being deterministically defined thus guaranteeing entanglement. See \cref{fig:line} for a comparison of different families of graph states. 
Notice that the plotted exact values of rank-width for grids and hypercubes (recall that $\ket{\cQ_k}$ can be prepared in CZ depth $k$, thus in depth logarithmic in the number of qubits) produce better constants for small numbers of qubits, even though their asymptotics are weaker.
\begin{figure}[ht]
\begin{center}
    \includegraphics[scale=0.85]{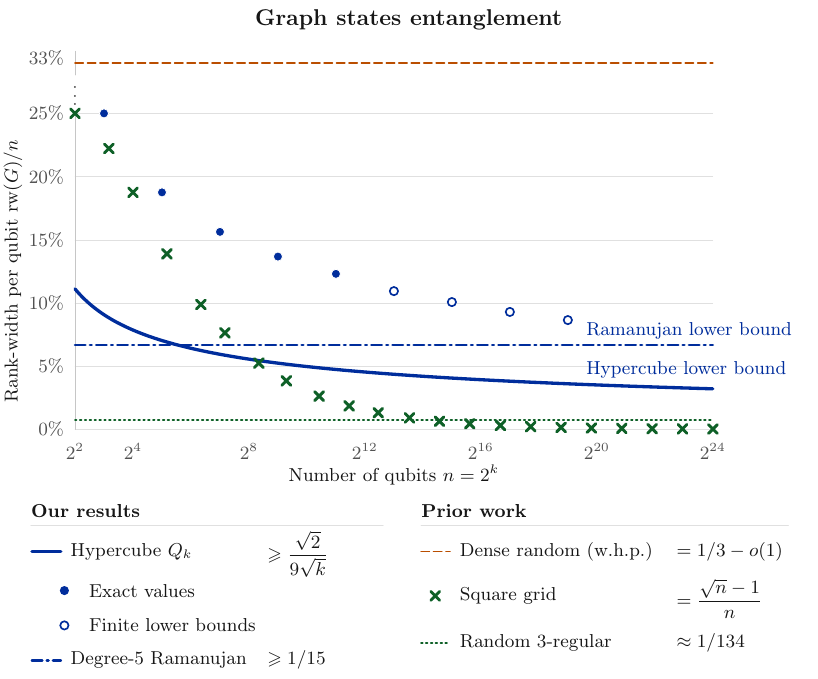}
    \caption{Normalized rank-width $\rw(G)/n$ versus the number of qubits.
Blue curves show our proved deterministic lower bounds: hypercube from Eq.~\eqref{hypercube}; $5$-regular Ramanujan graphs from \cref{thm:ram}.
Filled and open blue circles show exact hypercube values ($k=3,5,7,9,11$) and finite lower bounds ($k=13,15,17,19$), respectively (see \ifcompactlayout full version \cite{CGLBM26}\else Appendix~\ref{app:hypercube-values}\fi{} for values).
Green crosses give the exact square-grid value from \cite{Jelinek10}; the dotted green line reproduces the coefficient $1/(27\cdot4.95)\approx1/134$ from Eq.~\eqref{prior-3-reg}, achieved w.h.p. by prior work.
The orange line shows the upper bound on rank-width that $G(n,p)$ converges to w.h.p. by \cite{LLO12}.}
    \label{fig:line}
\end{center}
    \vspace{-5mm}
\end{figure}

These deterministic preparations can also serve as starting points for quantum advantage sampling experiments: adding local non-Clifford gates preserves their certified entanglement before measurement. For Cartesian products, the following remark connects this approach to the cluster-state sampling framework of Haferkamp et al. \cite{haferkamp2020closing}, with classical sampling hardness conditional on that of the corresponding cluster experiment.



\begin{remark}\label{anticoncentration}
For Cartesian products of connected graphs of bounded order, our fixed graph-state preparation supports a sampling scheme with at most one additional $T$ layer and local Clifford gates. If the corresponding eight-phase grid distributions are classically hard to sample given the phase choices, then so are these product distributions.

Choose one edge in each of the $k$ factors to obtain an induced $Q_k$. For sufficiently large $k$, this cube contains an induced rectangular grid on $\Theta(2^{k/2})$ vertices, with the aspect ratio required by Haferkamp et al.~\cite{haferkamp2020closing}. The bounded factor orders ensure that the total number of qubits is polynomial in the grid order.

On each grid qubit, independently and uniformly choose $j\in\{0,\ldots,7\}$ and apply $T^j$, where $T=\operatorname{diag}(1,e^{i\pi/4})$. 
These local gates preserve the entanglement certified by our rank-width bounds before measurement.

Measure the induced grid qubits in the $X$ basis and discard all other unneeded qubits by measuring them out in the $Z$ basis, retaining all outcomes. The $Z$ outcomes determine local $Z$ corrections on the grid~\cite[Proposition~1]{hein2004multiparty}, which only flip the corresponding $X$ outcomes. A classical algorithm sampling all these outcomes could therefore simulate the grid experiment: use the $Z$ outcomes to correct the grid bits, then retain only those corrected bits. This processing cannot increase total variation error.

Our deterministically prepared states therefore combine provably large entanglement width with a concrete route to quantum advantage experiments: at most one additional $T$ layer and local Clifford gates suffice for classically hard approximate sampling, assuming average-case $\#\mathrm{P}$-hardness of multiplicatively approximating the eight-phase grid probabilities.
\end{remark}

\subsection{Relation to other width parameters}\label{sec:other}

First, it turns out that our method, whenever using the strong chromatic index $\chistrong$ (an edge-colouring parameter analogous to the acyclic chromatic index but producing partitions into induced matchings instead) coincides precisely with computing a lower bound on another width parameter: (Maximum Induced Matching)-width or \emph{mim-width}, introduced by Vatshelle \cite{vatshelle2012new}. 
Interestingly, a witness of large mim-width is also a witness of rank-width. Thus, for all $G$, $\rw(G)\geq \mimw(G)$.
This allows us to reformulate \cref{thm:algo} in terms of strong edge colouring:
$
    \mimw(G) ~\ge~ \min_{p\in[\frac{1}{3},\frac{1}{2}]}\frac{\sF\pare{p\,|V|}}{\chistrong(G)}.
$
This gives the general asymptotic lower bound $\mimw(G)\geq \frac{\min |\partial X|}{\chistrong(G)}\geq \frac{\min |\partial X|}{2\Delta(G)^2}$ by the naive greedy upper bound on strong chromatic index. Thus, bounds obtained through strong edge colouring also give mim-width lower bounds; the directed acyclic matchings used for Cartesian products do not justify this transfer.

It is interesting to note the inequalities\footnote{where $\pw(G)$ is the path-width of $G$, and $\cw(G)$ is the clique-width of $G$.} for all $G$:
$$\mimw(G)\leq \rw(G)\leq \tw(G)+1 \leq \mathrm{pw}(G)+1~~\text{and}~~\rw(G)\leq \cw(G).$$
Thus, our rank-width lower bounds also apply to tree-width, path-width, and clique-width. We can also use any upper bound on a parameter greater than the rank-width (asymptotically). 
For instance, for fixed $k\ge2$, the $k$-dimensional grid $G_{k,m}$ of side length $m$ has $\rw(G_{k,m})=\Theta(m^{k-1})$: \cref{cor:cartesian} gives the lower bound, and a path decomposition using consecutive coordinate sections gives $\tw(G_{k,m})=O(m^{k-1})$.
    
Chandran, Kavitha and Subramanian \cite{CKS03} 
gave the asymptotic value of the tree-width for the $k$-dimensional hypercube $\cQ_k$ with $n = 2^k$ vertices: $\tw(\cQ_k) = \Theta(2^k/\sqrt{k}\,) = \Theta(n/\sqrt{\log{n}}\,)$. 
By \cite[Theorem~21]{FOT10}, for graphs excluding some fixed $K_{t,t}$ as a subgraph, tree-width (and clique-width as well)
is polynomially bounded by the rank-width, namely $\tw(G) = O(\rw(G)^t)$ for all $G$. Since $\cQ_k$ excludes $K_{3,3}$ as a subgraph, we get that $\rw(\cQ_k) = \Omega(\tw(\cQ_k)^{\nicefrac{1}{3}}) =  \Omega(n^{\nicefrac{1}{3} - o(1)})$, which was to the best of our knowledge the best known lower bound~\cite{dMOO24}.

Our lower bound Eq.~\eqref{hypercube} implies $\rw(\cQ_k)=\Theta(2^k/\sqrt{k})$. 


\subsection{Final remarks}

Interestingly, reminiscent of Jel\'inek's work on the rank-width of the square grid \cite{Jelinek10}, our lower bound on the rank-width of Ramanujan graphs used a matching-based technique to witness the asymptotic saturation of rank-width. By \cite{vatshelle2012new}, there exist many examples of graphs of mim-width $\Theta(1)$ and rank-width $\Theta(\sqrt{n}\,)$. The matchings returned by our recursive method are not necessarily induced. 
Whether the maximum induced matching width of $Q_k$ also has asymptotic order $2^k/\sqrt{k}$ remains open, although a maximal influence argument based on the KKL theorem already gives a lower bound in $\Omega((2^k\log k)/k)$.
This raises the question of what is the best parameter for \emph{expansion-to-rank-width} conversion. 

Furthermore notice that often the isoperimetric inequalities used hold for any $X\subsetneq V(G)$, including very small subsets. These bounds capture small set expansion but these often are extremal cases that degrade the bound.
If one uses a finer definition of edge expansion $h_{\mathrm{balanced}}=\min_{n/3\leq|X|\leq n/2}|\partial X|/|X|$ 
 only based on vertex subsets of balanced size, then \cref{thm:algo} can be strengthened (although if the graph studied is vertex-transitive for instance, a small poorly expanding set cannot be exceptional: symmetry reproduces it throughout the graph). 

In \cref{thm:cartesian}, our recursive construction only selects at most one edge of each Cartesian copy of a graph factor. For bounded-order factors, it turns out to already construct asymptotically optimally large directed-acyclic matchings, as witnessed by the path-width upper bounds for instance. However, if we allow the orders of the factors themselves to be unbounded, can the recursive construction further exploit the structure of each factor copy to improve the dependence on the factor orders?


Finally, the Ramanujan graphs we studied have no four-cycle which is crucial for our acyclicity argument, but makes the induced-grid sampling scheme remark incompatible with these graph states. 
Ghosh et al.~\cite{GHH25} establish anticoncentration for random regular graph states, averaging over both the graph and the phase choices. Extending this result to our fixed families would provide a starting point for proving sampling hardness directly for these graph states, without first extracting an induced grid.

\paragraph{AI disclosure.} After submission, the authors interacted with GPT 6 Astra to help explore algorithmic leads and their combinatorial analysis yielding the improved lower bounds for \cref{thm:ram,thm:cartesian}, the latter matching the bound they initially conjectured. The authors have rewritten these proofs and take full responsibility for the manuscript. The authors also used this model to proofread.








\label{boundary:references}\bibliographystyle{splncs04}
\bibliography{biblio}

@article{alon1986eigenvalues,
  author = {Alon, Noga},
  title = {Eigenvalues and expanders},
  journal = {Combinatorica},
  volume = {6},
  number = {2},
  pages = {83--96},
  year = {1986},
  doi = {10.1007/BF02579166},
}

@article{alon2020explicitexpandersdegreesize,
  author = {Alon, Noga},
  title = {Explicit expanders of every degree and size},
  journal = {Combinatorica},
  volume = {41},
  number = {4},
  pages = {447--463},
  year = {2021},
  doi = {10.1007/s00493-020-4429-x},
}

@book{BH11,
  author = {Brouwer, Andries E. and Haemers, Willem H.},
  title = {Spectra of graphs},
  series = {Universitext},
  publisher = {Springer},
  address = {New York},
  year = {2012},
  doi = {10.1007/978-1-4614-1939-6},
  isbn = {978-1-4614-1938-9},
}

@book{Bollobas01,
  author = {Bollob{\'a}s, B{\'e}la},
  title = {Random Graphs},
  edition = {2nd},
  series = {Cambridge Studies in Advanced Mathematics},
  volume = {73},
  publisher = {Cambridge University Press},
  year = {2001},
  doi = {10.1017/CBO9780511814068},
  isbn = {978-0-521-79722-1},
}

@article{bollobas1991edge,
  author = {Bollob{\'a}s, B{\'e}la and Leader, Imre},
  title = {Edge-isoperimetric inequalities in the grid},
  journal = {Combinatorica},
  volume = {11},
  number = {4},
  pages = {299--314},
  year = {1991},
  doi = {10.1007/BF01275667},
}

@article{BR18,
  author = {Baste, Julien and Rautenbach, Dieter},
  title = {Degenerate matchings and edge colorings},
  journal = {Discrete Applied Mathematics},
  volume = {239},
  pages = {38--44},
  year = {2018},
  doi = {10.1016/j.dam.2018.01.002},
}

@TECHREPORT{CGLBM26,
     AUTHOR = {Cam, Tristan and Gavoille, Cyril and Le~Borgne, Yvan and Martiel, Simon},
      TITLE = {Entanglement from Expansion: High Rank-Width in Deterministic Graphs},
INSTITUTION = {arXiv},
     NUMBER = {2606.07110v1 [cs.DM]},
      MONTH = jun,
       YEAR = 2026,
        DOI = {10.48550/arXiv.2606.07110}
}

@inproceedings{CKS03,
  author = {Chandran, L. Sunil and Kavitha, Telikepalli and Subramanian, C. R.},
  editor = {Warnow, Tandy and Zhu, Binhai},
  title = {Isoperimetric Inequalities and the Width Parameters of Graphs},
  booktitle = {Computing and Combinatorics},
  series = {Lecture Notes in Computer Science},
  volume = {2697},
  pages = {385--393},
  publisher = {Springer},
  address = {Berlin, Heidelberg},
  year = {2003},
  doi = {10.1007/3-540-45071-8_39},
}

@misc{collares2025counting,
  author = {Collares, Maur{\'\i}cio and Erde, Joshua and Geisler, Anna and Kang, Mihyun},
  title = {Counting independent sets in expanding bipartite regular graphs},
  year = {2025},
  doi = {10.48550/arXiv.2503.22255},
  eprint = {2503.22255},
  archiveprefix = {arXiv},
  howpublished = {arXiv preprint},
}

@misc{davies2025preparing,
  author = {Davies, James and Jena, Andrew},
  title = {Preparing graph states forbidding a vertex-minor},
  year = {2025},
  doi = {10.48550/arXiv.2504.00291},
  eprint = {2504.00291},
  archiveprefix = {arXiv},
  howpublished = {arXiv preprint},
}

@article{DEK24,
  author = {Do, Tuan Anh and Erde, Joshua and Kang, Mihyun},
  title = {A note on the width of sparse random graphs},
  journal = {Journal of Graph Theory},
  volume = {106},
  number = {2},
  pages = {273--295},
  year = {2024},
  doi = {10.1002/jgt.23081},
}

@article{diskin2024isoperimetric,
  author = {Diskin, Sahar and Erde, Joshua and Kang, Mihyun and Krivelevich, Michael},
  title = {Isoperimetric inequalities and supercritical percolation on high-dimensional graphs},
  journal = {Combinatorica},
  volume = {44},
  number = {4},
  pages = {741--784},
  year = {2024},
  doi = {10.1007/s00493-024-00089-0},
}

@unpublished{dMOO24,
  author = {de Oliveira, Mateus},
  title = {On the Rank-Width of the Hypercube},
  year = {2024},
  month = apr,
  note = {Personal communication},
}

@article{flammia2011direct,
  author = {Flammia, Steven T. and Liu, Yi-Kai},
  title = {Direct Fidelity Estimation from Few {Pauli} Measurements},
  journal = {Physical Review Letters},
  volume = {106},
  number = {23},
  pages = {230501},
  year = {2011},
  doi = {10.1103/physrevlett.106.230501},
}

@article{FOT10,
  author = {Fomin, Fedor V. and Oum, Sang-il and Thilikos, Dimitrios M.},
  title = {Rank-width and tree-width of {H}-minor-free graphs},
  journal = {European Journal of Combinatorics},
  volume = {31},
  number = {7},
  pages = {1617--1628},
  year = {2010},
  doi = {10.1016/j.ejc.2010.05.003},
}

@article{GHH25,
  author = {Ghosh, Soumik and Hangleiter, Dominik and Helsen, Jonas},
  title = {Random Regular Graph States Are Complex at Almost Any Depth},
  journal = {PRX Quantum},
  volume = {6},
  number = {4},
  pages = {040344},
  year = {2025},
  doi = {10.1103/52xz-3hpc},
}

@article{haferkamp2020closing,
  author = {Haferkamp, Jonas and Hangleiter, Dominik and Bouland, Adam and Fefferman, Bill and Eisert, Jens and Bermejo-Vega, Juani},
  title = {Closing gaps of a quantum advantage with short-time {Hamiltonian} dynamics},
  journal = {Physical Review Letters},
  volume = {125},
  number = {25},
  pages = {250501},
  year = {2020},
  doi = {10.1103/PhysRevLett.125.250501},
}

@article{Harper66,
  author = {Harper, Lawrence Hueston},
  title = {Optimal Numberings and Isoperimetric Problems on Graphs},
  journal = {Journal of Combinatorial Theory},
  volume = {1},
  number = {3},
  pages = {385--393},
  year = {1966},
  doi = {10.1016/S0021-9800(66)80059-5},
}

@article{hein2004multiparty,
  author = {Hein, Marc and Eisert, Jens and Briegel, Hans J.},
  title = {Multiparty entanglement in graph states},
  journal = {Physical Review A},
  volume = {69},
  number = {6},
  pages = {062311},
  year = {2004},
  doi = {10.1103/PhysRevA.69.062311},
}

@article{Jelinek10,
  author = {Jel{\'{i}}nek, V{\'{i}}t},
  title = {The rank-width of the square grid},
  journal = {Discrete Applied Mathematics},
  volume = {158},
  number = {7},
  pages = {841--850},
  year = {2010},
  doi = {10.1016/j.dam.2009.02.007},
}

@article{jozsa2003role,
  author = {Jozsa, Richard and Linden, Noah},
  title = {On the role of entanglement in quantum-computational speed-up},
  journal = {Proceedings of the Royal Society of London. Series A: Mathematical, Physical and Engineering Sciences},
  volume = {459},
  number = {2036},
  pages = {2011--2032},
  year = {2003},
  doi = {10.1098/rspa.2002.1097},
}

@article{kumabe2024complexity,
  author = {Kumabe, Soh and Mori, Ryuhei and Yoshimura, Yusei},
  title = {Complexity of graph-state preparation by {Clifford} circuits},
  journal = {Quantum},
  volume = {10},
  pages = {2165},
  year = {2026},
  doi = {10.22331/q-2026-07-18-2165},
  issn = {2521-327X},
  eprint = {2402.05874},
  archiveprefix = {arXiv},
  primaryclass = {quant-ph},
}

@article{LLO12,
  author = {Lee, Choongbum and Lee, Joonkyung and Oum, Sang-il},
  title = {Rank-width of random graphs},
  journal = {Journal of Graph Theory},
  volume = {70},
  number = {3},
  pages = {339--347},
  year = {2012},
  doi = {10.1002/jgt.20620},
}

@article{lubotzky1988ramanujan,
  author = {Lubotzky, Alexander and Phillips, Ralph and Sarnak, Peter},
  title = {{Ramanujan} graphs},
  journal = {Combinatorica},
  volume = {8},
  number = {3},
  pages = {261--277},
  year = {1988},
  doi = {10.1007/BF02126799},
}

@misc{martiel2026sampling,
  author = {Martiel, Simon and Chung, Jay-U and Seif, Alireza and Ghosh, Soumik and Hincks, Ian and Deshpande, Abhinav and Manabe, Hidetaka and Gu, Hanfeng and Pan, Feng and Fefferman, Bill and Gambetta, Jay M. and Javadi-Abhari, Ali},
  title = {Sampling hard circuits with verifiably high fidelity},
  year = {2026},
  doi = {10.48550/arXiv.2607.25941},
  eprint = {2607.25941},
  archiveprefix = {arXiv},
  primaryclass = {quant-ph},
  howpublished = {arXiv preprint},
}

@inproceedings{mohanty2020explicit,
  author = {Mohanty, Sidhanth and O'Donnell, Ryan and Paredes, Pedro},
  title = {Explicit near-{Ramanujan} graphs of every degree},
  booktitle = {Proceedings of the 52nd Annual ACM SIGACT Symposium on Theory of Computing},
  pages = {510--523},
  publisher = {Association for Computing Machinery},
  year = {2020},
  doi = {10.1145/3357713.3384231},
}

@article{mohar1989isoperimetric,
  author = {Mohar, Bojan},
  title = {Isoperimetric numbers of graphs},
  journal = {Journal of Combinatorial Theory, Series B},
  volume = {47},
  number = {3},
  pages = {274--291},
  year = {1989},
  doi = {10.1016/0095-8956(89)90029-4},
}

@article{Mor94,
  author = {Morgenstern, Moshe},
  title = {Existence and explicit constructions of {$q+1$} regular {Ramanujan} graphs for every prime power {$q$}},
  journal = {Journal of Combinatorial Theory, Series B},
  volume = {62},
  number = {1},
  pages = {44--62},
  year = {1994},
  doi = {10.1006/jctb.1994.1054},
}

@article{nilli1991second,
  author = {Nilli, A.},
  title = {On the second eigenvalue of a graph},
  journal = {Discrete Mathematics},
  volume = {91},
  number = {2},
  pages = {207--210},
  year = {1991},
  doi = {10.1016/0012-365X(91)90112-F},
}

@inproceedings{o2019parameterization,
  author = {O'Gorman, Bryan},
  title = {Parameterization of tensor network contraction},
  booktitle = {14th Conference on the Theory of Quantum Computation, Communication and Cryptography (TQC 2019)},
  series = {Leibniz International Proceedings in Informatics (LIPIcs)},
  volume = {135},
  pages = {10:1--10:19},
  publisher = {Schloss Dagstuhl -- Leibniz-Zentrum f{\"u}r Informatik},
  year = {2019},
  doi = {10.4230/LIPIcs.TQC.2019.10},
  eprint = {1906.00013},
  archiveprefix = {arXiv},
}

@article{orus2014practical,
  author = {Or{\'u}s, Rom{\'a}n},
  title = {A practical introduction to tensor networks: Matrix product states and projected entangled pair states},
  journal = {Annals of Physics},
  volume = {349},
  pages = {117--158},
  year = {2014},
  doi = {10.1016/j.aop.2014.06.013},
}

@article{Oum08,
  author = {Oum, Sang-il},
  title = {Rank-width is less than or equal to branch-width},
  journal = {Journal of Graph Theory},
  volume = {57},
  number = {3},
  pages = {239--244},
  year = {2008},
  doi = {10.1002/jgt.20280},
}

@article{oum2006approximating,
  author = {Oum, Sang-il and Seymour, Paul D.},
  title = {Approximating clique-width and branch-width},
  journal = {Journal of Combinatorial Theory, Series B},
  volume = {96},
  number = {4},
  pages = {514--528},
  year = {2006},
  doi = {10.1016/j.jctb.2005.10.006},
}

@article{perez2006matrix,
  author = {Perez-Garcia, David and Verstraete, Frank and Wolf, Michael M. and Cirac, J. Ignacio},
  title = {Matrix product state representations},
  journal = {Quantum Information and Computation},
  volume = {7},
  number = {5\&6},
  pages = {401--430},
  year = {2007},
  doi = {10.26421/QIC7.5-6-1},
}

@article{tillich2000edge,
  author = {Tillich, Jean-Pierre},
  title = {Edge isoperimetric inequalities for product graphs},
  journal = {Discrete Mathematics},
  volume = {213},
  number = {1--3},
  pages = {291--320},
  year = {2000},
  doi = {10.1016/S0012-365X(99)00189-2},
}

@article{van2006universal,
  author = {Van den Nest, Maarten and Miyake, Akimasa and D{\"u}r, Wolfgang and Briegel, Hans J.},
  title = {Universal resources for measurement-based quantum computation},
  journal = {Physical Review Letters},
  volume = {97},
  number = {15},
  pages = {150504},
  year = {2006},
  doi = {10.1103/PhysRevLett.97.150504},
}

@phdthesis{vatshelle2012new,
  author = {Vatshelle, Martin},
  title = {New width parameters of graphs},
  school = {University of Bergen},
  year = {2012},
  isbn = {9788230820988},
  url = {https://hdl.handle.net/1956/6166},
}

@article{vidal2003efficient,
  author = {Vidal, Guifr{\'e}},
  title = {Efficient classical simulation of slightly entangled quantum computations},
  journal = {Physical Review Letters},
  volume = {91},
  number = {14},
  pages = {147902},
  year = {2003},
  doi = {10.1103/PhysRevLett.91.147902},
}


\ifcompactlayout\else
\newpage
\appendix
\renewcommand{\theHsection}{appendix.\arabic{section}}

\iffullversion\else
\section{Deferred proofs}
\label{app:deferred-proofs}
\ifbriefprelim
\subsection{The cut degree matching lemma}
\label{app:proof-cut-degree-matching}
We give the complete proof of \cref{lemma:cut-degree-matching}.
\begin{proof}
We start with the whole cut-graph $C = G[\partial X]$, including its vertices of cut degree $\Delta$ which do not belong to $S$, and record the selected edges in $M$.  At each step, we select an edge and delete its endpoints together with some neighbours, retaining all crossing edges between surviving vertices. We count only deletions from $S$: although a step may delete $\Delta+1$ vertices in total, we will show that at most $\Delta$ of them belong to $S$. 
The set $S$ remains fixed by the original cut degrees. During the procedure, neighbours and degrees refer to the graph remaining after these deletions. 
The initial cut-graph has no isolated vertices. The two rules below preserve this property, ensuring that every vertex of $S$ is deleted in a step selecting an edge. While an edge remains, apply one of the following rules, using a fixed arbitrary vertex ordering to resolve choices.

\emph{1. If the remaining part of $\oX$ contains a vertex $v$ of degree one}, let $u\in X$ be its only neighbour. Select $(u,v)$ in $M$ and delete $u$ together with all degree-one vertices of $\oX$ adjacent to $u$ (including $v$). Any surviving neighbour of $u$ had degree at least two and therefore still has a neighbour. The deleted vertices of $\oX$ were adjacent only to $u$, so no surviving vertex of $X$ loses a neighbour.

\emph{2. Otherwise, every remaining vertex of $\oX$ has degree at least two.} Choose an edge $(u,v)$ with $u\in X$ and $v\in\oX$, select it in $M$, and delete $v$ together with all its neighbours.  Before this step, every remaining vertex of $\oX$ distinct from $v$ shared at most one neighbour with $v$, since two common neighbours would give a four-cycle. Each of them therefore loses at most one neighbour and remains nonisolated. Surviving vertices of $X$ were not adjacent to $v$, so they lose no neighbour.

Neither rule creates an isolated vertex, so every vertex of $S$ is eventually deleted. Each step deletes one vertex ($u$ in the first rule, $v$ in the second) and some of its neighbours: at most $\Delta+1$ vertices. If a step deletes $\Delta+1$ vertices, then $u$ in the first rule, or $v$ in the second, has degree $\Delta$ immediately before the step. Deletion cannot increase degrees, and the original cut-graph had maximum degree at most $\Delta$, so its original cut degree was also $\Delta$. This vertex did not belong to $S$, whose vertices have original cut degree at most $\Delta-1$, by definition. Such a step therefore deletes at most $\Delta$ vertices of $S$. Any smaller step deletes at most $\Delta$ vertices in total. Thus every selected edge accounts for at most $\Delta$ deleted vertices of $S$, giving $|S|\le\Delta|M|$ and hence $|M|\ge \frac{|S|}{\Delta}$.

The selected edges form a matching because their endpoints are deleted immediately. In the first rule, the target $v\in\oX$ has $u$ as its only remaining neighbour, which gets deleted. In the second rule, all its remaining neighbours are deleted. Any endpoint $w\in X$ of an edge selected later by the procedure is present when $(u,v)$ is selected. If $w$ were adjacent to $v$ in $G$, either rule would delete $w$. Thus no earlier selected $v$ is adjacent to a later selected endpoint in $X$. Every arc in $A_{M,X}$ therefore goes from a later selected edge to an earlier one, so $A_{M,X}$ is acyclic.\end{proof}
\subsection{The degree-five Ramanujan bound}
\label{app:proof-ramanujan-five}
We give the complete proof of \cref{prop:ramanujan-five}.
To apply \cref{lemma:cut-degree-matching}, we need many boundary vertices with cut degree between one and four. Since $G$ is $5$-regular, being incident to a crossing edge ensures cut degree at least one, while having an incident edge staying within the same cut side ensures cut degree at most four. We therefore count pairs of edges sharing a vertex, one crossing the cut and one staying within a side. These pairs are exactly the paths of length two with endpoints on opposite sides of the cut (with the middle vertex in either side), but no alternating path since endpoints would be on the same cut side. The proof below uses a spectral estimate of this count to show that many vertices have cut degree between one and four.

\begin{proof}[\cref{prop:ramanujan-five}]
Fix a cut $(X,\oX)$, its cut-graph $C=G[\partial X]$, and let $S=\{v\in V(C):1\le \deg_{C}(v)<5\}$. Write $N_G(v)$ for the set of neighbours of $v$.
All adjacency eigenvalues other than five have absolute value at most four, so the neighbour-count inequality of Brouwer and
Haemers \cite[Proposition~4.3.1]{BH11} gives
\begin{equation}\label{eq:ramanujan-five-spectral}
 \sum_{v\in V(G)} \left(|N_G(v)\cap X|-\frac{5|X|}{n}\right)^2 \le16\frac{|X||\oX|}{n}.
\end{equation}
Here $|N_G(v)\cap X|$ equals $\deg_{C}(v)$ when $v\in\oX$, and $5-\deg_{C}(v)$ when $v\in X$.

Each vertex of $X$ belongs to $N_G(v)\cap X$ for exactly five vertices $v$, namely its five neighbours. Thus
\begin{equation}\label{identity-X}    
 \sum_{v\in V(G)}|N_G(v)\cap X|=5|X|\,.
\end{equation}

We now expand each square and sum its three terms. The constant term occurs once for each of the $n$ vertices, giving
\begin{align}\label{expansion-square}
 &\sum_v\left(|N_G(v)\cap X|-\frac{5|X|}{n}\right)^2\notag\\
 \quad=&\sum_v|N_G(v)\cap X|^2-\frac{10|X|}{n}\sum_v|N_G(v)\cap X|+n\left(\frac{5|X|}{n}\right)^2\notag\\
 \quad=&\sum_v|N_G(v)\cap X|^2-\frac{50|X|^2}{n}+\frac{25|X|^2}{n}\notag\\
 \quad=&\sum_v|N_G(v)\cap X|^2-\frac{25|X|^2}{n}\,.
\end{align}

Rearranging Eq.~\eqref{expansion-square}, we get
\begin{equation}\label{expansion-square2}
 \sum_v|N_G(v)\cap X|^2 =\frac{25|X|^2}{n} +\sum_v\left(|N_G(v)\cap X|-\frac{5|X|}{n}\right)^2\,.
\end{equation}

We now return to counting paths of length two with endpoints on opposite sides of the cut. For a fixed middle vertex $v$, we choose one of its $\deg_C(v)$ neighbours across the cut and one of its $5-\deg_C(v)$ neighbours on the same side. Summing over middle vertices counts each such path once, so their total number satisfies
\begin{align}\label{path-count} 
 \sum_v \deg_{C}(v)\bigl(5-\deg_{C}(v)\bigr) &=\sum_v |N_G(v)\cap X|\bigl(5-|N_G(v)\cap X|\bigr)\notag\\
 &=5\sum_v|N_G(v)\cap X|-\sum_v|N_G(v)\cap X|^2\notag\\
 &=25|X|-\sum_v|N_G(v)\cap X|^2
\end{align}
where the last equality uses Eq.~\eqref{identity-X}.

Substituting Eq.~\eqref{expansion-square2} into the path count Eq.~\eqref{path-count}, and using
$n-|X|=|\oX|$, we obtain
\begin{align}
 \sum_v \deg_{C}(v)\bigl(5-\deg_{C}(v)\bigr)
 &=25|X|-\frac{25|X|^2}{n}
   -\sum_v\left(|N_G(v)\cap X|-\frac{5|X|}{n}\right)^2\notag\\
 &=\frac{25|X||\oX|}{n}
   -\sum_v\left(|N_G(v)\cap X|-\frac{5|X|}{n}\right)^2\notag\\
 &\ge\frac{25|X||\oX|}{n}-\frac{16|X||\oX|}{n}=\frac{9|X||\oX|}{n}.
 \label{eq:ramanujan-five-pairs}
\end{align}
The last inequality is Eq.~\eqref{eq:ramanujan-five-spectral}.


Every such path has its middle vertex in $S$ by construction. There are only two cases for the number of neighbours a vertex of $S$ has on each cut side: either one and four, or two and three. So it is the middle vertex of at most $\max\{1\cdot4,\,2\cdot3\}=6$ such paths. Thus
\[
 6|S|\ge\sum_v \deg_{C}(v)\bigl(5-\deg_{C}(v)\bigr)
 \ge\frac{9|X||\oX|}{n}.
\]
The cut-graph has maximum degree at most five and no four-cycle.
Applying \cref{lemma:cut-degree-matching} with $\Delta=5$ (taking the cut-graph cannot increase the maximum degree of the graph) and \cref{lemma:jelinek} gives 
$\,\rho_G(X)\ge\frac{|S|}{5}
 \ge\frac{3|X||\oX|}{10n}.
$\end{proof}
\fi
\subsection{The Cartesian-product theorem}
\label{app:proof-cartesian}
We complete the construction given in the main text by proving
acyclicity and the matching estimate used for vertex expansion.
\ifbriefprelim
\begin{figure}[ht]
 \centering
 \includegraphics[width=.85\linewidth]{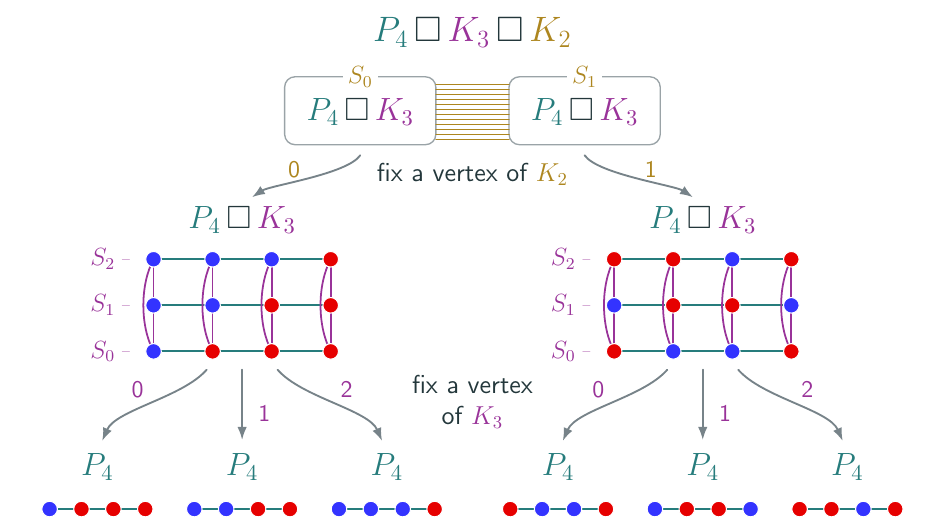}
 \caption{Recursive sections for the Cartesian product of  $G_1=P_4$, $G_2=K_3$ and $G_3=K_2$. Fixing the last coordinate gives two sections isomorphic to $G_1\square G_2$; fixing the $G_2$ coordinate next gives six sections, each corresponding to a $G_1$, that inherit the fixed cut and to which the base case is applied. Blue and red indicate the inherited cut sides $X$ and $\oX$. The top boxes schematically represent the sections obtained by fixing a bit for the last coordinate; on the first step a candidate matching will be formed from the gold edges shown between them. Arrows show recursive calls to form the second candidate from edges within both sections; the two candidates are compared upon return.}
 \label{fig:recursion}
\end{figure}
\fi
\begin{proof}[Continuation of the proof]
\paragraph{Acyclicity.}
We prove by induction on $k$ that, for every cut $(X,\oX)$, both orientation classes of the matching returned by the construction are directed acyclic.

Every edge of a Cartesian product changes exactly one coordinate. Orient each selected edge from $A_i$ to $B_i$ in the coordinate it changes. The \emph{height} of a product vertex is the number of coordinates $i$ whose entries lie in $B_i$. Every selected edge raises height by one and is called ascending, whereas every edge of $G$ changes height by at most one. The two orientation classes consist of the selected edges ascending from $X$ to $\oX$, and those ascending from $\oX$ to $X$.

For $k=1$, the returned matching has at most one edge, so the acyclicity claim holds.
For $k\ge2$, assume that the claim holds for every cut of a product of $k-1$ factors. Let $\widehat M$ be either orientation class. By exchanging $X$ and $\oX$ if necessary, assume that its edges ascend from $X$ to $\oX$. For distinct edges $e_0=(u_0,v_0)$ and $e_1=(u_1,v_1)$ of $\widehat M$, with $u_0,u_1\in X$ and $v_0,v_1\in\oX$, recall that by definition $A_{\widehat M,X}$ has an arc $e_0\to e_1$ exactly when $(v_0,u_1)$ is an edge of $G$.

Suppose first that $M=M_{\mathrm{between}}$. Recall that each copy $C_u$ fixes the first $k-1$ coordinates and lets only coordinate $k$ vary. Since the construction selects at most one edge in each copy, $e_0$ and $e_1$ lie in different copies. These copies have different fixed $(k-1)$-tuples, so $v_0$ and $u_1$ differ in at least one of the first $k-1$ coordinates.
They also differ in coordinate $k$\,: both selected edges are oriented from $A_k$ to $B_k$, so the $k$-th coordinate of $v_0$ lies in $B_k$, whereas that of $u_1$ lies in $A_k$. Thus they differ in at least two coordinates and cannot be adjacent in $G$.
Hence $A_{\widehat M,X}$ has no arcs. In particular, $\widehat M$ is an induced matching in the cut-graph $G[\partial X]$ and is directed acyclic.

Suppose now that $M=M_{\mathrm{inside}}$ and that $A_{\widehat M,X}$ contains a directed cycle. By the definition of the edge-adjacency digraph, the cycle describes the following closed walk in the original graph: traverse a \emph{selected} ascending edge, then a \emph{returning} edge to the lower endpoint of the next selected edge, and repeat. The walk contains equally many selected and returning edges. Each selected edge raises height by one, while each returning edge lowers it by at most one. Since the walk is closed, every returning edge must lower height by exactly one.
A selected $M_{\mathrm{inside}}$ edge leaves coordinate $k$ fixed. A returning edge that changes the last coordinate must move it from $B_k$ to $A_k$, since it must lower height. To close the walk and return to the starting vertex coordinate, a later step would have to move the last coordinate back from $A_k$ to $B_k$. This is impossible: selected edges leave coordinate $k$ fixed, and a returning edge moving it from $A_k$ to $B_k$ would raise height. Thus no step of the closed walk can change coordinate $k$.

Every vertex of the walk therefore lies in the same section $S_v$. Within this section, $M_{\mathrm{inside}}$ consists precisely of the edges returned by the recursive call, namely $M_v$. The selected edges of the walk therefore belong to $M_v$ and still ascend from $X\cap S_v$ to $\oX\cap S_v$. Every returning edge also lies within $S_v$, so it still gives an arc from one selected edge to the next in the recursive edge-adjacency digraph. Thus the same directed cycle occurs in one orientation class of $M_v$. But $S_v$ is a product of $k-1$ factors, and the induction hypothesis says that both orientation classes returned by this recursive call are directed acyclic, reaching a contradiction.

Thus both orientation classes of $M$ are directed acyclic. Choose $\widehat M$ to be the larger class. \cref{lemma:jelinek} then gives
$\, \rho_G(X)\ge |\widehat M|\ge\frac{|M|}{2}$.

\paragraph{Counting.}
Let $P=|X||\oX|$ count pairs of vertices on opposite cut sides, whether adjacent in $G$ or non-edges. For $k\ge2$, in each section $S_v$ put
$x_v=|X\cap S_v|$ and $p_v=x_v\left(\frac{n}{n_k}-x_v\right)$. Thus $p_v$ counts pairs with one vertex in $X\cap S_v$ and the other in $\oX\cap S_v$. 

\noindent Between two distinct sections $S_v$ and $S_w$, the number of opposite-side pairs is
\begin{equation}\label{eq:ce-two-sections}
 \begin{aligned}
 x_v\left(\frac n{n_k}-x_w\right) &+x_w\left(\frac n{n_k}-x_v\right)
 =\frac{n}{n_k}(x_v+x_w)-2x_vx_w\\
 &=\left[\frac{n}{n_k}(x_v+x_w)-x_v^2-x_w^2\right]+\left[x_v^2+x_w^2 -2x_vx_w\right]\\
 &=p_v+p_w+(x_v-x_w)^2\,.
 \end{aligned}
\end{equation}

To get all types of pairs in $P$, we have to sum over all unordered pairs $\{v,w\}$ of distinct section labels, and add the pairs lying within individual sections:

\begin{equation}
 \begin{aligned}
 P &=\sum_v p_v+\sum_{\{v,w\}}\bigl(p_v+p_w+(x_v-x_w)^2\bigr)\\
 &=\sum_v p_v+(n_k-1)\sum_v p_v+\sum_{\{v,w\}}(x_v-x_w)^2\\
 &=n_k\sum_v p_v+\sum_{\{v,w\}}(x_v-x_w)^2\,.
 \end{aligned}
\end{equation}
Indeed, each $p_v$ appears exactly $n_k-1$ times in the sum over unordered pairs of distinct sections, once for each other section. Together with its contribution from within $S_v$, its total coefficient is therefore $n_k$. We write $P=A+B$\,:
\begin{equation}\label{eq:ce-pairs}
 P=\underbrace{n_k\sum_vp_v}_{A}
   +\underbrace{\sum_{\{v,w\}}(x_v-x_w)^2}_{B}.
\end{equation}
We will control $A$ by recursion in the sections and $B$ by the mixed
coordinate copies. We now prove the following hypothesis by induction, to relate $P$ to $|M|$\,: 
\begin{equation}\label{eq:ce-induction}
 8P^2\le n^2|M|^2\sum_{i=1}^k n_i^2\,.
 \end{equation}
We induct on $k$. Having either side of the cut empty gives $P=0$. For $k=1$ and both sides nonempty, $n=n_1$ and our algorithm gives $|M|=1$. Since $P$ is a quadratic and is maximal when $|X|=|\oX|=\frac{n}{2}$, we have $P\le \frac{n^2}{4}$ and $8P^2\le \frac{n^4}{2}\le n^2n_1^2|M|^2$ indeed holds.

Suppose $k\ge2$ and assume Eq.~\eqref{eq:ce-induction} holds for every cut of $G_1\square\cdots\square G_{k-1}$. We begin by bounding the first term $A$ of Eq.~\eqref{eq:ce-pairs} using recursion within the sections. Each $S_v$ induces a product of $G_1,\ldots,G_{k-1}$ of order $\frac{n}{n_k}$, with pair count $p_v$ and returned matching $M_v$. The induction hypothesis applied to $(X\cap S_v,\oX\cap S_v)$ then gives:
\begin{equation}\label{eq:ce-section-hypothesis}
 8p_v^2\le\left(\frac n{n_k}\right)^2|M_v|^2
             \sum_{i=1}^{k-1}n_i^2\,.
\end{equation}
Taking square roots gives
\begin{equation}\label{eq:ce-section-induction}
 \sqrt8\,p_v\le\frac n{n_k}
       \sqrt{\sum_{i=1}^{k-1}n_i^2}\,|M_v|\,.
\end{equation}
The factor $\frac n{n_k}\sqrt{\sum_{i=1}^{k-1}n_i^2}$ does not depend on the section label $v$. Adding the section-dependent inequality Eq.~\eqref{eq:ce-section-induction} over $v\in V(G_k)$ therefore gives
\begin{equation}\label{eq:ce-section-sum}
 \sqrt8\sum_vp_v
 \le\frac n{n_k}\sqrt{\sum_{i=1}^{k-1}n_i^2}\,
       \underbrace{\sum_v|M_v|}_{|M_{\mathrm{inside}}|\,\le\,|M|}.
\end{equation}
Multiply by $n_k$ to recover $A=n_k\sum_vp_v$, and use the
candidate comparison Eq.~\eqref{eq:ce-candidates}:
\begin{equation}\label{eq:ce-recover-A}
 \sqrt8\,A\le n\sqrt{\sum_{i=1}^{k-1}n_i^2}\,|M|\,.
\end{equation}
Squaring both sides gives
\begin{equation}\label{eq:ce-inside}
 8A^2\le n^2|M|^2\sum_{i=1}^{k-1}n_i^2\,.
\end{equation}
This accounts for the first $k-1$ factors in the induction hypothesis inequality Eq.~\eqref{eq:ce-induction}. To include the last factor $G_k$, it remains to prove
\begin{equation}\label{eq:ce-missing}
 8P^2-8A^2\le n^2n_k^2|M|^2.
\end{equation}

To bound the differences, fix distinct sections $S_v,S_w$ and suppose $x_v\ge x_w$. Pair each vertex $(u,v)\in X\cap S_v$ with its corresponding vertex $(u,w)\in S_w$, as in \cref{fig:ce-surplus}. At most $x_w$ partners belong to $X$ by definition. Thus at least $x_v-x_w$ belong to $\oX$, each in a distinct mixed copy $C_u$. Each such copy contributes one edge to $M_{\mathrm{between}}$ by construction. Exchanging $v,w$ if necessary, this gives using Eq.~\eqref{eq:ce-candidates}
\begin{equation}\label{eq:ce-surplus}
 |x_v-x_w|\le |M_{\mathrm{between}}|\le |M|.
\end{equation}

\begin{figure}[htbp]
 \centering
 \includegraphics[width=.90\linewidth]{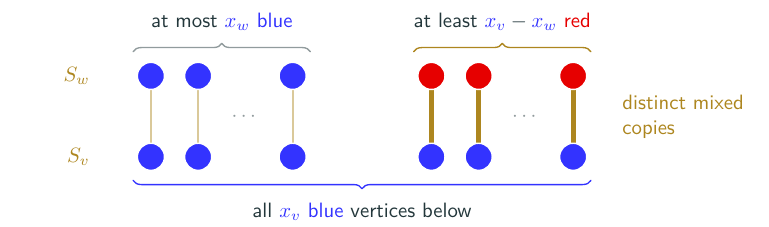}
 \caption{Minimum number of opposite-side pairs between $S_v$ and $S_w$, for $x_v\ge x_w$. Blue means $X$ and red means $\oX$. Only the $x_v$ blue vertices of $S_v$ and their corresponding pair partners in $S_w$ are shown. Segments indicate pairing, not necessarily graph edges. Assuming $x_v\ge x_w$, in the worst case, all of the blue vertices of the top $S_w$ section are only paired with blue vertices in the bottom $S_v$ section. This still forces at least $x_v-x_w$ pairs (the rest of the pairs containing a blue vertex in the bottom section) to be a an opposite-side pair.}
 \label{fig:ce-surplus}
\end{figure}

To bound the second term $B$ of Eq.~\eqref{eq:ce-pairs}, apply the bound Eq.~\eqref{eq:ce-surplus} to one factor in each square:
\begin{equation}\label{eq:ce-linear}
 B=\sum_{\{v,w\}}(x_v-x_w)^2\le |M_{\mathrm{between}}|\sum_{\{v,w\}}|x_v-x_w|\,.
\end{equation}
Let $Q$ count all pairs with one vertex in $X$, the other in $\oX$, and both vertices in the same coordinate copy $C_u$. 
Since all copies are pairwise disjoint, we simply get $Q=\sum_u |X\cap C_u||\oX\cap C_u|$. Fix two distinct sections $S_v,S_w$, and pair $(u,v)$ with $(u,w)$
for each $u$. As in the argument preceding Eq.~\eqref{eq:ce-linear}, at least
$|x_v-x_w|$ of these pairs have their vertices on opposite sides
of the cut. Each such pair is counted by $Q$, since its vertices
belong to the same coordinate copy $C_u$. Conversely, every pair
counted by $Q$ arises from exactly one unordered pair of sections.
Summing over $\{v,w\}$ therefore gives

\begin{equation}\label{eq:ce-surplus-pairs}
 \sum_{\{v,w\}}|x_v-x_w|\le Q\,.
\end{equation}

In each mixed $C_u$, the two cut-side sizes sum to $n_k$, so again their product is at most $n_k^2/4$. There are exactly $|M_{\mathrm{between}}|$ mixed copies by definition, so we can also upper-bound $Q$:
\begin{equation}\label{eq:ce-copy-pairs}
Q\le \frac{n_k^2}{4}|M_{\mathrm{between}}|\,.
\end{equation}
Insert Eq.~\eqref{eq:ce-surplus-pairs} and Eq.~\eqref{eq:ce-copy-pairs} into Eq.~\eqref{eq:ce-linear}, and use the candidate
comparison Eq.~\eqref{eq:ce-candidates}:
\begin{equation}\label{eq:ce-between}
 B\le |M_{\mathrm{between}}|Q
 \le\frac{n_k^2|M_{\mathrm{between}}|^2}{4}
 \le\frac{n_k^2|M|^2}{4}.
\end{equation}

We now bound the remaining difference $8P^2-8A^2$ in Eq.~\eqref{eq:ce-missing}. Since $P-A=B$ by the pair-counting identity Eq.~\eqref{eq:ce-pairs}:
\begin{equation}\label{eq:ce-difference}
 P^2-A^2=(P-A)(P+A)=B(P+A).
\end{equation}
It remains to bound $P+A$. The identity $P=A+B$ gives $0\le A\le P$. Also recall that $P\le \frac{n^2}{4}$, therefore
\begin{equation}\label{eq:ce-remaining-factor}
 P+A\le 2P\le\frac{n^2}{2}.
\end{equation}
Insert the bounds Eq.~\eqref{eq:ce-between} and Eq.~\eqref{eq:ce-remaining-factor} into the difference-of-squares identity Eq.~\eqref{eq:ce-difference}, and multiply by $8$:
\begin{equation}\label{eq:ce-remaining-contribution}
 \begin{aligned}
 8P^2-8A^2
 &=8B(P+A)\\
 &\le 8\cdot\frac{n_k^2|M|^2}{4}\cdot\frac{n^2}{2}
  =n^2n_k^2|M|^2.
 \end{aligned}
\end{equation}
This is the remaining contribution in Eq.~\eqref{eq:ce-missing}. Adding it to the bound we obtained for the first term $A$ using the first $k-1$ factors in Eq.~\eqref{eq:ce-inside}, we get
\begin{equation}\label{eq:ce-finish}
 \begin{aligned}
 8P^2
 &=8A^2+8B(P+A)\\
 &\le n^2|M|^2\sum_{i=1}^{k-1}n_i^2+n^2n_k^2|M|^2\\
 &=n^2|M|^2\sum_{i=1}^k n_i^2.
 \end{aligned}
\end{equation}
This proves our induction hypothesis in Eq.~\eqref{eq:ce-induction}. Taking square roots and rearranging yields
\begin{equation}\label{eq:ce-rank-conversion}
 |M|\ge\frac{\sqrt8\,P}{n\sqrt{\sum_{i=1}^k n_i^2}}.
\end{equation}
By the acyclicity of either direction choice, every cut $(X,\oX)$ therefore satisfies
\begin{equation}\label{eq:ce-cut-conclusion}
 \rho_G(X)\ge\frac{|M|}{2}
 \ge\frac{\sqrt2\,|X||\oX|}{n\sqrt{\sum_{i=1}^k n_i^2}}.
\end{equation}
For a balanced cut, $|X||\oX|\ge\frac{2n^2}{9}$, thus \cref{lemma:rw} gives
\begin{equation}\label{eq:ce-rw-conclusion}
 \rw(G)\ge
 \frac{\sqrt2\cdot(2n^2/9)}{n\sqrt{\sum_{i=1}^k n_i^2}}
 =\frac{2\sqrt2}{9}\cdot\frac{n}{\sqrt{\sum_{i=1}^k n_i^2}}.
\end{equation}\end{proof}
\fi

\subsection{The size of the outer boundary}\label{app:question}

Recall from \cref{dekk-answer} that Diskin, Erde, Kang and Krivelevich \cite[Question~7.3(1)]{diskin2024isoperimetric} ask whether, for a fixed integer $C>1$, every Cartesian product $G=\square_{i=1}^kG_i$ of connected $d_i$-regular graphs with $2\le |V(G_i)|\le C$ satisfies
\[
 |N_G(X)|=\Omega\!\left(\frac{|X|}{\sqrt d}\right)
 \qquad\text{for every }X\subseteq V(G)\text{ with }|X|\le \frac{n}{2},
\]
where $n=|V(G)|$, $d=\sum_i d_i$, and $N_G(X)$ is the set of vertices in $\overline{X}=V(G)\setminus X$ with a neighbour in $X$, often called the \emph{outer boundary} of $X$, closely related to the study of vertex expansion (see Eq.~\eqref{vertex-exp}).

The matching estimate Eq.~\eqref{eq:ce-verification-matching} in the proof of \cref{thm:cartesian} gives an affirmative answer.
A matching across $(X,\overline{X})$ has distinct endpoints in $N_G(X)$.
Thus, using $|X|\le \frac{n}{2}$ and $|V(G_i)|\le C$, we obtain from Eq.~\eqref{eq:ce-verification-matching}

$$|M| \ge \frac{\sqrt8\,|X||\oX|}{n\sqrt{\sum_{i=1}^k n_i^2}} \ge \frac{2\sqrt2\,|X|(n-|X|)}{Cn\sqrt{k}} \ge \frac{\sqrt2}{C}\cdot\frac{|X|}{\sqrt{k}}$$
\begin{equation}
    |\partial X|\ge|N_G(X)| \ge |M|  \ge \frac{\sqrt2}{C}\cdot\frac{|X|}{\sqrt{k}}\ge \frac{\sqrt2}{C}\cdot\frac{|X|}{\sqrt d}
\end{equation}

Indeed, each of the $k$ factors is connected and has at least two vertices, hence $d_i\ge1$. Therefore $d=\sum_{i=1}^k d_i\ge k$, which gives the last inequality.
Notice that the bound in terms of $k$ does not require regularity, only connectivity.\footnote{For a related expansion estimate and induction in the case of regular bipartite factors, see prior work from Collares, Erde, Geisler and Kang \cite[Theorem~1.2(c)]{collares2025counting}.}

The exponent $1/2$ is optimal: in $Q_k$, for odd $k$, the words of weight at most $(k-1)/2$ form a set of size $2^{k-1}$ whose outer boundary has size $\binom{k}{(k+1)/2}=\Theta(2^k/\sqrt k)$, and $d=k$.

Note that Eq.~\iffullversion\eqref{eq:ce-rank-conversion}\else\eqref{eq:ce-verification-matching}\fi{} also implies the following edge and vertex expansion bounds for Cartesian products of graphs

\begin{equation}\label{vertex-exp}
    h(G)=\min_{0<|X|\le\frac{n}{2}}\frac{|\partial X|}{|X|} \ge \min_{0<|X|\le\frac{n}{2}}\frac{|N_G(X)|}{|X|}\ge\frac{\sqrt2}{\sqrt{\sum_{i=1}^kn_i^2}}
\end{equation}

Moreover,  if for a fixed integer $C\ge2$, all $n_i\le C$, then
    $$h(G)\ge \min_{0<|X|\le\frac{n}{2}}\frac{|N_G(X)|}{|X|}\ge\frac{\sqrt2}{C\sqrt{k}}$$

These bounds are useful to lower bound the size of separators, for instance, which must disconnect a vertex set $X$ from all its outer neighbours.

\section{Hypercube rank-width values}
\label{app:hypercube-values}

Table~\ref{tab:finite-hypercube-values} lists the computed numerical values used for the (non asymptotic) hypercube points in \cref{fig:line}. Equalities denote exact values; inequalities denote lower bounds. The plotted values are normalized
by the number of vertices $n=2^k$.

\begin{table}[htbp]
    \centering
    \caption{Exact values and finite lower bounds for the rank-width of $Q_k$.}
    \label{tab:finite-hypercube-values}
    \begin{tabular}{|r r c|}
        \hline
        $k$ & $n=2^k$ & $\rw(Q_k)$ \\
        \hline
        ~~ 3 ~~&       ~~8~~ & ~~$=2$~~ \\
        ~~ 5 ~~&      ~~32~~ & ~~$=6$~~ \\
        ~~ 7 ~~&     ~~128~~ & ~~$=20$~~ \\
        ~~ 9 ~~&     ~~512~~ & ~~$=70$~~ \\
        ~~ 11 ~~& ~~$2\,048$~~ & ~~$=252$~~ \\
        \hline
        ~~ 13 ~~& ~~$8\,192$~~ & ~~$\ge896$~~ \\
        ~~ 15 ~~& ~~$32\,768$~~ & ~~$\ge3\,300$~~ \\
        ~~ 17 ~~& ~~$131\,072$~~ & ~~$\ge12\,177$~~ \\
        ~~ 19 ~~& ~~$524\,288$~~ & ~~$\ge45\,331$~~ \\
        \hline
    \end{tabular}
\end{table}

\fi
\end{document}